\documentclass{fundam}
\usepackage[utf8]{inputenc}
\usepackage[T1]{fontenc}

\usepackage{amssymb}
\DeclareMathSymbol{\shortminus}{\mathbin}{AMSa}{"39}
\hypersetup{breaklinks=true}
\usepackage{dashrule}
\usepackage{hyphenat}
\usepackage{mathrsfs}
\makeatletter
\long\def\@makecaption#1#2{%
        \vskip\abovecaptionskip
        \sbox\@tempboxa{#1: #2}%
        \ifdim \wd\@tempboxa >\hsize
        #1: #2\par
        \else
        \global \@minipagefalse
        \hb@xt@\hsize{\hfil\box\@tempboxa\hfil}%
        \fi
        \vskip\belowcaptionskip}
\makeatother

\usepackage[style=base,width=.85\textwidth]{caption}
\usepackage{longtable}
\usepackage{booktabs}
\usepackage[inline]{enumitem}

\usepackage{algorithm}
\usepackage[noend]{algpseudocode}

\makeatletter
\renewcommand\thealgorithm{\thesection.\arabic{algorithm}}
\let\c@algorithm\c@equation
\newenvironment{breakablealgorithm}
{
\begin{center}
        \refstepcounter{algorithm}
        \hrule height.8pt depth0pt \kern2pt
        \renewcommand{\caption}[2][\relax]{
                {\raggedright\textbf{\ALG@name~\thealgorithm} ##2\par}%
                \ifx\relax##1\relax 
                \addcontentsline{loa}{algorithm}{\protect\numberline{\thealgorithm}##2}%
                \else 
                \addcontentsline{loa}{algorithm}{\protect\numberline{\thealgorithm}##1}%
                \fi
                \kern2pt\hrule\kern2pt
        }
}{
\kern2pt\hrule\relax%
\end{center}
}
\makeatother
\renewcommand*\Call[2]{\textproc{#1}(#2)}
\usepackage{graphicx}
\usepackage{pgf}
\usepackage{tikz}
\usetikzlibrary{matrix,decorations.pathmorphing,decorations.markings,decorations.pathreplacing}

\usepackage[group-minimum-digits=4]{siunitx}
\usepackage{hyphsubst}
\HyphSubstLet{english}{usenglishmax}
\usepackage[english]{babel}
\usepackage{csquotes}
\usepackage{mathtools}
\usepackage[numbers,sort&compress]{natbib}
\numberwithin{equation}{section}

\newcommand{\mscript}[1]{{$\scriptscriptstyle #1$}}

\def\ov#1{\overline{#1}}
\def\wt#1{\widetilde{#1}}

\newcommand{\ab}{\allowbreak}
\newcommand{\cG}{{\check G}}
\renewcommand{\AA}{\mathbb{A}}
\newcommand{\DD}{\mathbb{D}}
\newcommand{\EE}{\mathbb{E}}
\newcommand{\ZZ}{\mathbb{Z}}
\newcommand{\NN}{\mathbb{N}}
\newcommand{\QQ}{\mathbb{Q}}
\newcommand{\RR}{\mathbb{R}}
\newcommand{\CCC}{\mathbb{C}}
\newcommand{\MM}{\mathbb{M}}
\newcommand{\CN}{\mathcal{N}}
\newcommand{\CH}{\mathcal{H}}
\newcommand{\cox}{\mathrm{cox}}
\newcommand{\Cox}{\mathrm{Cox}}
\newcommand{\CI}{\mathcal{I}}
\DeclareMathOperator{\crk}{\mathbf{crk}}

\DeclareMathOperator{\specc}{{\mathbf{specc}}}
\DeclareMathOperator{\Ker}{Ker}
\DeclareMathOperator{\lcm}{lcm}

\newcommand{\bh}{\mathbf{h}}
\newcommand{\bc}{\mathbf{c}}
\newcommand{\Gl}{\mathrm{Gl}}
\newcommand{\Dyn}{\mathrm{Dyn}}
\newcommand{\Gln}{\Gl(n;\,\ZZ)}
\newcommand{\sgn}{\mathrm{sgn}}
\newcommand*{\inlineeqno}[1]{\refstepcounter{equation}\hfill (\theequation)\label{#1}}

\newcommand{\mppps}{\raisebox{.5pt}{\scalebox{0.75}{$\scriptstyle +$}}}
\newcommand{\mppms}{\scalebox{0.75}[1]{$\scriptstyle -$}}
\newcommand{\mpppss}{\raisebox{.375pt}{\scalebox{0.75}{$\scriptscriptstyle +$}}}
\newcommand{\mppmss}{\scalebox{0.75}[1]{$\scriptscriptstyle -$}}
\newcommand{\mmm}{\scalebox{0.75}[1]{$-$}}
\newcommand{\ppp}{\raisebox{0.682pt}{\scalebox{0.75}{$+$}}}
\newcommand{\eqdef}{ \coloneqq}
\newtheorem{conjecture}[equation]{Conjecture}
\newtheorem{problem}[equation]{Problem}
\newcommand{\tikzsetnextfilename}[1]{}%
\newcommand{\tikzexternaldisable}{}%
\newcommand{\tikzexternalenable}{}%
\newenvironment{proofof}[1]
{\trivlist\PRstyle\item[]{\bfseries #1:}\newline}{\QED\endtrivlist}
\def\squareforqed{\hbox{\rlap{$\sqcap$}$\sqcup$}}
\def\QED{\ifmmode\squareforqed\else{\unskip\nobreak\hfil
                \penalty50\hskip1em\null\nobreak\hfil\squareforqed
                \parfillskip=0pt\finalhyphendemerits=0\endgraf}\fi}

\newcommand{\grapheAn}[2]{\tikzsetnextfilename{grapheAn}\begin{tikzpicture}[label distance=-2pt, xscale=#1, yscale=#2]
                \node[circle, fill=black, inner sep=0pt, minimum size=3.5pt, label=above:\mscript{1}] (n1) at (0  , 0  ) {};
                \node[circle, fill=black, inner sep=0pt, minimum size=3.5pt, label=above:\mscript{2}] (n2) at (1  , 0  ) {};
                \node (n3) at (2  , 0  ) {\mscript{}};
                \node (n4) at (3  , 0  ) {\mscript{}};
                \node[circle, fill=black, inner sep=0pt, minimum size=3.5pt, label=above:\mscript{n-1}] (n5) at (4  , 0  ) {};
                \node[circle, fill=black, inner sep=0pt, minimum size=3.5pt, label=above:\mscript{n}] (n6) at (5  , 0  ) {};
                
                \draw [-, shorten <= -2.50pt, shorten >= 2.50pt] (n4) to  (n5);
                \draw [dotted, -, shorten <= -2.50pt, shorten >= -2.50pt] (n3) to  (n4);
                \foreach \x/\y in {1/2, 5/6}
                \draw [-, shorten <= 2.50pt, shorten >= 2.50pt] (n\x) to  (n\y);
                \draw [-, shorten <= 2.50pt, shorten >= -2.50pt] (n2) to  (n3);
\end{tikzpicture}}

\newcommand{\grapheDn}[2]{\tikzsetnextfilename{grapheDn}\begin{tikzpicture}[label distance=-2pt, xscale=#1, yscale=#2]
                \node[circle, fill=black, inner sep=0pt, minimum size=3.5pt, label=above:\mscript{1}] (n1) at (0  , 0  ) {};
                \node[circle, fill=black, inner sep=0pt, minimum size=3.5pt, label=right:\mscript{2}] (n2) at (1  , 0.6) {};
                \node[circle, fill=black, inner sep=0pt, minimum size=3.5pt, label=above right:\mscript{3}] (n3) at (1  , 0  ) {};
                \node (n4) at (2  , 0  ) {\mscript{}};
                \node (n5) at (3  , 0  ) {\mscript{}};
                \node[circle, fill=black, inner sep=0pt, minimum size=3.5pt, label=above:\mscript{n-1}] (n6) at (4  , 0  ) {};
                \node[circle, fill=black, inner sep=0pt, minimum size=3.5pt, label=above:\mscript{n}] (n7) at (5  , 0  ) {};
                
                \draw [-, shorten <= -2.50pt, shorten >= 2.50pt] (n5) to  (n6);
                \draw [dotted, -, shorten <= -2.50pt, shorten >= -2.50pt] (n4) to  (n5);
                \foreach \x/\y in {1/3, 2/3, 6/7}
                \draw [-, shorten <= 2.50pt, shorten >= 2.50pt] (n\x) to  (n\y);
                \draw [-, shorten <= 2.50pt, shorten >= -2.50pt] (n3) to  (n4);
\end{tikzpicture}}

\newcommand{\grapheEsix}[2]{\tikzsetnextfilename{grapheEsix}\begin{tikzpicture}[label distance=-2pt, xscale=#1, yscale=#2]
                \node[circle, fill=black, inner sep=0pt, minimum size=3.5pt, label=above:\mscript{1}] (n1) at (0  , 0  ) {};
                \node[circle, fill=black, inner sep=0pt, minimum size=3.5pt, label=above:\mscript{2}] (n2) at (1  , 0  ) {};
                \node[circle, fill=black, inner sep=0pt, minimum size=3.5pt, label=above left:\mscript{3}] (n3) at (2  , 0  ) {};
                \node[circle, fill=black, inner sep=0pt, minimum size=3.5pt, label=right:\mscript{4}] (n4) at (2  , 0.6  ) {};
                \node[circle, fill=black, inner sep=0pt, minimum size=3.5pt, label=above:\mscript{5}] (n5) at (3  , 0  ) {};
                \node[circle, fill=black, inner sep=0pt, minimum size=3.5pt, label=above:\mscript{6}] (n6) at (4  , 0  ) {};
                
                \foreach \x/\y in {1/2, 2/3, 3/5, 4/3, 5/6}
                \draw [-, shorten <= 2.50pt, shorten >= 2.50pt] (n\x) to  (n\y);
\end{tikzpicture}}

\newcommand{\grapheEseven}[2]{\tikzsetnextfilename{grapheEseven}\begin{tikzpicture}[label distance=-2pt, xscale=#1, yscale=#2]
                \node[circle, fill=black, inner sep=0pt, minimum size=3.5pt, label=above:\mscript{1}] (n1) at (0  , 0  ) {};
                \node[circle, fill=black, inner sep=0pt, minimum size=3.5pt, label=above:\mscript{2}] (n2) at (1  , 0  ) {};
                \node[circle, fill=black, inner sep=0pt, minimum size=3.5pt, label=above:\mscript{3}] (n3) at (2  , 0  ) {};
                \node[circle, fill=black, inner sep=0pt, minimum size=3.5pt, label=above right:\mscript{4}] (n4) at (3  , 0) {};
                \node[circle, fill=black, inner sep=0pt, minimum size=3.5pt, label=right:\mscript{5}] (n5) at (3  , 0.6  ) {};
                \node[circle, fill=black, inner sep=0pt, minimum size=3.5pt, label=above:\mscript{6}] (n6) at (4  , 0  ) {};
                \node[circle, fill=black, inner sep=0pt, minimum size=3.5pt, label=above:\mscript{7}] (n7) at (5  , 0  ) {};
                
                \foreach \x/\y in {1/2, 2/3, 3/4, 4/5, 4/6, 6/7}
                \draw [-, shorten <= 2.50pt, shorten >= 2.50pt] (n\x) to  (n\y);
\end{tikzpicture}}

\newcommand{\grapheEeight}[2]{\tikzsetnextfilename{grapheEeight}\begin{tikzpicture}[label distance=-2pt, xscale=#1, yscale=#2]
                \node[circle, fill=black, inner sep=0pt, minimum size=3.5pt, label=above:\mscript{1}] (n1) at (0  , 0  ) {};
                \node[circle, fill=black, inner sep=0pt, minimum size=3.5pt, label=above:\mscript{2}] (n2) at (1  , 0  ) {};
                \node[circle, fill=black, inner sep=0pt, minimum size=3.5pt, label=above right:\mscript{3}] (n3) at (2  , 0  ) {};
                \node[circle, fill=black, inner sep=0pt, minimum size=3.5pt, label=right:\mscript{4}] (n4) at (2  , 0.6) {};
                \node[circle, fill=black, inner sep=0pt, minimum size=3.5pt, label=above:\mscript{5}] (n5) at (3  , 0  ) {};
                \node[circle, fill=black, inner sep=0pt, minimum size=3.5pt, label=above:\mscript{6}] (n6) at (4  , 0  ) {};
                \node[circle, fill=black, inner sep=0pt, minimum size=3.5pt, label=above:\mscript{7}] (n7) at (5  , 0  ) {};
                \node[circle, fill=black, inner sep=0pt, minimum size=3.5pt, label=above:\mscript{8}] (n8) at (6  , 0  ) {};
                
                \foreach \x/\y in {1/2, 2/3, 3/5, 4/3, 5/6, 6/7, 7/8}
                \draw [-, shorten <= 2.50pt, shorten >= 2.50pt] (n\x) to  (n\y);
\end{tikzpicture}}

\newcommand{\graphOnePeakpan}{
\tikzsetnextfilename{hasseposit0Anmi}\begin{tikzpicture}[baseline={([yshift=1.75pt]current bounding box)},label distance=-2pt, xscale=1]
        \node[circle, fill=black, inner sep=0pt, minimum size=3.5pt, label=below:$\scriptstyle 1$] (n1) at (0  , 0  ) {};
        \node[circle, fill=black, inner sep=0pt, minimum size=3.5pt, label=below:$\scriptstyle 2$] (n2) at (1  , 0  ) {};
        \node (n3) at (2  , 0  ) {};
        \node (n4) at (3  , 0  ) {};
        \node[circle, fill=black, inner sep=0pt, minimum size=3.5pt, label=below:$\scriptstyle n-1$] (n5) at (4  , 0  ) {};
        \node[circle, fill=black, inner sep=0pt, minimum size=3.5pt, label=below:$\scriptstyle \phantom{1}n\phantom{1}$] (n6) at (5  , 0  ) {};
        
        \draw [-stealth, shorten <= -2.50pt, shorten >= 2.50pt] (n4) to  (n5);
        \draw [dotted, -, shorten <= -2.50pt, shorten >= -2.50pt] (n3) to  (n4);
        \foreach \x/\y in {1/2, 5/6}
        \draw [-stealth, shorten <= 2.50pt, shorten >= 2.50pt] (n\x) to  (n\y);
        \draw [-stealth, shorten <= 2.50pt, shorten >= -2.50pt] (n2) to  (n3);
\end{tikzpicture}}

\newcommand{\graphOnePeakpanPrinc}{
\tikzsetnextfilename{posetprincpan}
\begin{tikzpicture}[baseline={([yshift=-2.75pt]current bounding box)},
        label distance=-2pt,xscale=0.6, yscale=0.5]
        \node[circle, fill=black, inner sep=0pt, minimum size=3.5pt, label=above left:$\scriptscriptstyle 1$] (n1) at (0  , 2  ) {};
        \node[circle, fill=black, inner sep=0pt, minimum size=3.5pt, label=below left:$\scriptscriptstyle 2$] (n2) at (0  , 0  ) {};
        \node[circle, fill=black, inner sep=0pt, minimum size=3.5pt, label=above:$\scriptscriptstyle 3$] (n3) at (1  , 2  ) {};
        \node[circle, fill=black, inner sep=0pt, minimum size=3.5pt, label=above:$\scriptscriptstyle 4$] (n4) at (2  , 2  ) {};
        \node (n5) at (3  , 2  ) {$\scriptscriptstyle $};
        \node (n6) at (4  , 2  ) {$\scriptscriptstyle $};
        \node[circle, fill=black, inner sep=0pt, minimum size=3.5pt, label=right:$\scriptscriptstyle n\mppmss 1$] (n7) at (1  , 1  ) {};
        \node[circle, fill=black, inner sep=0pt, minimum size=3.5pt, label=right:$\scriptscriptstyle n$] (n8) at (6  , 1  ) {};
        \node[circle, fill=black, inner sep=0pt, minimum size=3.5pt, label=above right:$\scriptscriptstyle p$] (n9) at (5  , 2  ) {};
        \node[circle, fill=black, inner sep=0pt, minimum size=3.5pt, label=below:$\scriptscriptstyle p\mpppss 1$] (n10) at (1  , 0  ) {};
        \node (n11) at (2  , 0  ) {$\scriptscriptstyle $};
        \node (n12) at (4  , 0  ) {$\scriptscriptstyle $};
        \node[circle, fill=black, inner sep=0pt, minimum size=3.5pt, label=below right:$\scriptscriptstyle n\mppmss 2$] (n13) at (5  , 0  ) {};
        \foreach \x/\y in {1/3, 1/7, 2/7, 2/10, 3/4, 9/8, 13/8}
        \draw [-stealth, shorten <= 2.50pt, shorten >= 2.50pt] (n\x) to  (n\y);
        \foreach \x/\y in {4/5, 10/11}
        \draw [-stealth, shorten <= 2.50pt, shorten >= -2.50pt] (n\x) to  (n\y);
        \foreach \x/\y in {5/6, 11/12}
        \draw [dotted, -, shorten <= -2.50pt, shorten >= -2.50pt] (n\x) to  (n\y);
        \foreach \x/\y in {6/9, 12/13}
        \draw [-stealth, shorten <= -2.50pt, shorten >= 2.50pt] (n\x) to  (n\y);
\end{tikzpicture}}

\makeatletter
\let\c@figure\c@equation
\let\c@table\c@equation
\let\c@definition\c@equation
\let\c@theorem\c@equation
\let\c@fact\c@equation
\let\c@lemma\c@equation
\let\c@example\c@equation
\let\c@assumption\c@equation
\let\c@proposition\c@equation
\let\c@remark\c@equation
\let\c@corollary\c@equation
\let\c@claim\c@equation
\let\ftype@table\ftype@figure 
\makeatother

\publyear{2022}
\papernumber{2060}
\volume{182}
\issue{1}

\begin{document}

\def\matitle{A Coxeter type classification of Dynkin type $\AA_n$ non-negative posets}
\title{\matitle}

\address{ul. Chopina 12/18, 87-100 Toruń, Poland}

\author{Marcin G\k{a}siorek\\
Faculty of Mathematics and Computer Science\\
Nicolaus Copernicus University\\
ul. Chopina 12/18, 87-100 Toruń, Poland\\
mgasiorek{@}mat.umk.pl} 

\maketitle

\runninghead{M. G\k{a}siorek}{\matitle}

\begin{abstract}
We continue the Coxeter spectral analysis of finite connected posets $I$ that are non-negative in the sense
that their symmetric Gram matrix $G_I:=\frac{1}{2}(C_I + C_I^{tr})\in\MM_{m}(\QQ)$ is positive semi-definite of rank $n\geq 0$, where $C_I\in\MM_m(\ZZ)$ is the incidence matrix of $I$ encoding the relation $\preceq_I$. We extend the results of [Fundam. Inform., 139.4(2015),
347--367] and give a complete Coxeter spectral classification of finite connected posets $I$ of Dynkin type $\AA_n$.

We show that such posets $I$, with $|I|>1$, yield  exactly $\lfloor\frac{m}{2}\rfloor$  Coxeter types, one of which describes the positive (i.e., with $n=m$) ones. We give an exact description and calculate the number of posets of every type.
Moreover, we prove that, given a pair of such posets $I$ and $J$, the incidence matrices $C_I$ and $C_J$ are $\ZZ$-congruent if and only if $\specc_I = \specc_J$,
and present deterministic algorithms that calculate a $\ZZ$-invertible matrix defining such a $\ZZ$-congruence in a polynomial time.

\end{abstract}

\section{Introduction}

Coxeter spectral study of finite posets and, more generally, edge-bipartite graphs, is inspired by their applications in the representation theory of posets, finite groups, classical orders, finite-dimensional
algebras over a field $K$, and cluster $K$-algebras, see \cite{ASS,bondarenkoSystemsSubspacesUnitary2013,Si92,simsonMeshGeometriesRoot2011}, and \cite{simsonWeylOrbitsMatrix2022,simsonCoxeterGramClassificationPositive2013,SimZaj_intmms,mrozCongruencesEdgebipartiteGraphs2016,gasiorekAlgorithmicStudyNonnegative2015,jimenezgonzalezGraphTheoreticalFramework2021}.
The algebraic framework of this study is given in \cite{simsonCoxeterGramClassificationPositive2013,simsonFrameworkCoxeterSpectral2013}.

We note that the Coxeter spectral classification of finite posets, up to the Gram $\ZZ$-congruences defined in Section~\ref{section:preliminaries}, grew up from many different branches of mathematics and computer science and is successfully applied in Lie theory, Diophantine geometry, algebraic combinatorics, representation theory, matrix analysis, graph theory, combinatorial and graph algorithms, singularity theory, and related areas. The reader is referred to \cite{simsonWeylOrbitsMatrix2022,simsonCoxeterGramClassificationPositive2013,zajacPolynomialTimeInflation2020} and \cite[Section 6.1.2]{gasiorekAlgorithmicCoxeterSpectral2020} for a more detailed discussion.

In the present paper, we study posets that are non-negative of Dynkin type $\AA_m$, as defined in Section~\ref{section:preliminaries}. In \cite{gasiorekStructureNonnegativePosets2022arxiv} we give a characterization of such posets in terms of their Hasse digraphs (see Fact~\ref{fact:a:hasse}). One of the main results of the present paper is Theorem~\ref{thm:mainthm:stronggram}, which gives a detailed description of such posets up to the strong Gram $\ZZ$-congruence (see Definition~\ref{df:congruences}). This yields a complete Coxeter spectral classification of this class of posets and generalizes the results of \cite{gasiorekAlgorithmicStudyNonnegative2015,gasiorekOnepeakPosetsPositive2012,gasiorekCoxeterTypeClassification2019}. In particular, we show that the complex  Coxeter spectrum $\specc_I\subseteq\CCC$ \eqref{eq:pos_cox_specc} determines a connected non-negative poset $I$ of Dynkin type $\Dyn_I=\AA_m$ uniquely, up to the strong Gram $\ZZ$-congruence.

\begin{theorem}\label{thm:mainthm:stronggram}
Let $I=(I,\preceq)$ be a finite connected non-negative poset of corank $\crk_I\geq 0$ and Dynkin type $\Dyn_I=\AA_{|I|-\crk_I}$.
\begin{enumerate}[label={\textnormal{(\alph*)}}]
\item\label{thm:mainthm:stronggram:dichotomy} Poset $I$ is either positive or principal, that is, $\crk_I\in\{0,1\}$.
\item\label{thm:mainthm:stronggram:posit} If $I$ is positive, then
\begin{enumerate}[label=\normalfont{(b\arabic*)}, leftmargin=4ex]
\item\label{thm:mainthm:stronggram:posit:opcongr} $I$ is strongly Gram $\ZZ$-congruent with a one peak poset ${}_0 \AA_{n-1}^*$ \eqref{hasse:pospositzeroanmi},
\item\label{thm:mainthm:stronggram:posit:coxpol} $\cox_I(t) = t^{n+1}+\ldots + t + 1\in\ZZ[t]$,
\item\label{thm:mainthm:stronggram:posit:coxnum} $\bc_I=n+1$.
\end{enumerate}
\item\label{thm:mainthm:stronggram:princ} If $I$ is principal, then
\begin{enumerate}[label=\normalfont{(c\arabic*)}, leftmargin=4ex]
\item\label{thm:mainthm:stronggram:princ:pancongr} $I$ is strongly Gram $\ZZ$-congruent with a canonical two peak poset ${}_p\wt \AA_{n}$ \eqref{hasse:pospprincpan}, where 
$p\eqdef c(I)$ is the cycle index of $I$ \eqref{eq:df:cycleindex},
\item\label{thm:mainthm:stronggram:princ:coxpol} $\cox_I(t) = t^n- t^p-t^{n-p}+1\in\ZZ[t]$,
\item\label{thm:mainthm:stronggram:princ:coxnum} $\bc_I=\infty$ and $\check \bc_I=\lcm(p,n-p)$.
\end{enumerate}
\end{enumerate}
\end{theorem}

In other words, there exist precisely $\lfloor\frac{n}{2}\rfloor$ Coxeter types of connected Dynkin type $\AA_n$ non-negative posets, and exactly one of them describes positive ones. In Theorem~\ref{thr:posnum:a} we give formulae for the exact number of all, up to the isomorphism, posets of a given Coxeter polynomial. 
\begin{theorem}\label{thr:posnum:a}
        Assume that $n$ is a natural number. Up to poset isomorphism, there exist exactly:
        \begin{enumerate}[label={\textnormal{(\alph*)}}]
                \item\label{thr:posnum:a:posit} $N(P_n)$ connected non-negative posets $I$ with a Coxeter polynomial $\cox_I(t) = t^{n+1}+\cdots + t + 1$, where 
                \begin{equation*}%
                        N(P_n)=
                        \begin{cases}
                                2^{n-2}, & \textnormal{if $n\geq 2$ is even},\\[0.1cm]
                                2^{\frac{n - 3}{2}} + 2^{n - 2}, & \textnormal{if $n\geq 1$ is odd,}\\
                        \end{cases}
                \end{equation*}
                and every such a poset $I$ is positive, i.e., $\crk_I=0$;
                \item\label{thr:posnum:a:printc} $N(C_n,p)-1$ connected non-negative posets $I$ with the Coxeter polynomial $\cox_I(t) = t^n- t^p-t^{n-p}+1$, where
                $2\leq \frac{n}{2}\leq p\leq n-2$,
                \begin{equation*}%
                        N(C_n,p)=
                        \begin{cases}
                                \frac{1}{n}\sum_{d|\gcd(n,p)}\varphi(d)\binom{n/d}{p/d}, & \textnormal{if }p\in\{\lceil\frac{n}{2} \rceil, \ldots, n-2 \}\textnormal{ and } p\neq\frac{n}{2},\\[0.1cm]
                                \frac{1}{2n}\sum_{d|\gcd(n,\frac{n}{2})}\varphi(d)\binom{n/d}{n/2d}+2^{\frac{n}{2}-2}, & \textnormal{if }p=\frac{n}{2}.\\
                        \end{cases}
                \end{equation*}
               where $\varphi$ is Euler's totient function. Moreover, every such a poset $I$ is principal, i.e., $\crk_I=1$.
        \end{enumerate}
\end{theorem}

As a direct result of Theorem~\ref{thm:mainthm:stronggram} we obtain a useful algebraic tool for checking whether two posets are strongly Gram $\ZZ$-congruent. Namely, every two  connected non-negative posets $I$ and $J$ of Dynkin type $\Dyn_I=\Dyn_J=\AA_m$ are strongly Gram $\ZZ$-congruent if and only if $\cox_I=\cox_J$.

\begin{corollary}\label{corr:main:dyncongr}
If $I$ and $J$ are non-negative connected posets of Dynkin type $\AA_m$, then $I\approx_\ZZ J$ if and only if $\cox_I(t)=\cox_J(t)$ or, equivalently, $\specc_I=\specc_J$.
\end{corollary}

This is a  solution to the following variant of the Coxeter spectral analysis problem formulated by Simson in \cite{simsonCoxeterGramClassificationPositive2013,simsonFrameworkCoxeterSpectral2013} and studied in \cite{gasiorekAlgorithmicCoxeterSpectral2020,gasiorekAlgorithmicStudyNonnegative2015,gasiorekCoxeterTypeClassification2019,gasiorekOnepeakPosetsPositive2012,gasiorekStructureNonnegativePosets2022arxiv,jimenezgonzalezCoxeterInvariantsNonnegative2022,jimenezgonzalezGraphTheoreticalFramework2021,jimenezgonzalezIncidenceGraphsNonnegative2018,mrozCongruencesEdgebipartiteGraphs2016,zajacPolynomialTimeInflation2020,zajacStructureLoopfreeNonnegative2019}, for a wide class of connected non-negative posets $I$ of Dynkin type $\AA_m$.

\begin{problem}\label{problem_main}
        When the Coxeter polynomial $\cox_I(t)\in\ZZ[t]$ (equivalently: the Coxeter spectrum $\specc_I\subseteq \CCC$) of a finite poset $I$
        determines the incidence matrix $C_I\in\MM_{n}(\ZZ)$ %
        uniquely, up to the strong Gram
        $\ZZ$-congruence?
\end{problem}

There is also a natural question of not only determining whether posets $I$ and $J$ are strong Gram $\ZZ$-congruent, but computing a matrix that defines this congruence, see \cite{simsonCoxeterGramClassificationPositive2013,simsonFrameworkCoxeterSpectral2013}. 

\begin{problem}\label{problem_alg}
Construct an algorithm that computes such a $\ZZ$-invertible matrix $B\in\Gln$, that $B^{tr}\cdot C_I\cdot B=C_J$,  for any pair of strong Gram $\ZZ$-congruent connected non-negative posets $I$ and $J$.
\end{problem}
We present a solution to Problem~\ref{problem_alg} for a wide class of connected non-negative posets, i.e., posets of Dynkin type $\AA_m$, in 
the form of efficient polynomial-time deterministic algorithms: Algorithm~\ref{alg:refl:posit} (for $\crk_I=0$) and Algorithm~\ref{alg:refl:princ} (for $\crk_I\neq 0$), see Section~\ref{sec:algorithms} for throughout description and complexity analysis.

\section{Preliminaries and notation}\label{section:preliminaries}
Throughout the paper, by $\NN\subseteq\ZZ\subseteq\QQ\subseteq \RR\subseteq \CCC$ we denote the set of non-negative integers, the ring of integers, the rational, the real, and the complex number fields, respectively.
We view $\ZZ^n$, with $n\geq 1$, as a free abelian group, and we denote by $e_1, \ldots, e_n$ the standard $\ZZ$-basis of $\ZZ^n$. We use a row notation for vectors $v=[v_1,\ldots,v_n]$ and we write $v^{tr}$ to denote a column vector. 
Given $n\geq 1$, by $\MM_n(\ZZ)$ we denote the $\ZZ$-algebra of all square $n$ by $n$ matrices, by $E\in \MM_n(\ZZ)$ the identity matrix, and by $\Gl(n,\ZZ):=\{A\in \MM_n(\ZZ), \,\det A\in \{-1, 1\}\}\subseteq\MM_n(\ZZ)$ the general integral $\ZZ$-linear group. 

We say that two square integral matrices $X\in\MM_n(\ZZ)$ and $Y\in\MM_n(\ZZ)$
are $\ZZ$-\textit{congruent} if there exists such a matrix $B\in\Gln$,
that $B^{tr}\cdot X\cdot B=Y$. We denote this relation by $X\sim_\ZZ Y$ and write $X\overset{B}{\sim_\ZZ} Y$to denote the matrix $B\in\MM_n(\ZZ)$ defining the congruence.
\medskip

By a finite \textbf{signed graph} we mean a triplet $G=(V_G, E_G, \sgn_G)$ consisting of a finite 
set of \textit{vertices} $V_G$, a finite set of \textit{edges} $E_G$ and a \textit{sign} function $\sgn_G\colon E_G\to\{-1,1\}$. By an \textit{edge}, we mean a pair
of (not necessarily distinct) vertices and, for simplicity of presentation, throughout the paper we assume that  $V_G=\{1,\ldots,n\}$. 
By a directed [signed] graph (\textit{digraph} or \textit{quiver}) we mean a graph $G$, whose edges $\alpha\in E_G$ have designated source $s(\alpha)\in V_G$ and target $t(\alpha)\in V_G$. Following \cite{simsonCoxeterGramClassificationPositive2013,simsonFrameworkCoxeterSpectral2013}, by a finite \textit{edge-bipartite graph} (bigraph) we mean a signed graph $\Delta=(\Delta_0, \Delta_1, \sgn \colon \Delta_1\mapsto \{+1,-1\})$, with the sign map constant on the multiset $\Delta_1(u,v)=\Delta_1(v,u)\subseteq\Delta_1$ of edges adjacent with the vertices $u,v\in\Delta_0$.
Graphically, we represent bigraphs as graphs with multiple edges, where:
\begin{itemize}
        \item \textit{positive} edges 
        $\Delta_1^+:= \{e\in\Delta_1;\, \sgn(e)=+1 \}$
are denoted by dotted lines $u\,\hdashrule[3pt]{22pt}{0.4pt}{1pt}v$ and
        \item \textit{negative} edges 
        $\Delta_1^- := \{e\in\Delta_1;\, \sgn(e)=-1 \}$
         are denoted by full lines $u\,\rule[3pt]{22pt}{0.4pt}v$.
\end{itemize}
We note that graphs $G=(V_G,E_G)$ can be viewed as bigraphs $\Delta=(V_G,E_G,\sgn)$ with a constant sign function $\sgn(e)\eqdef -1$ for every $e\in E_{G}$. 

Two (di)graphs $G=(V_{G},E_{G})$ and $G'=(V_{G'},E_{G'})$ are called \textbf{isomorphic} $G\simeq G'$ if there exists a bijection $f\colon V_G\to V_{G'}$ that preserves (directed) edges. By \textbf{underlying graph} $\ov D$, we mean a graph obtained from signed digraph $D$ by forgetting the orientation and signs of its edges. We call graph $G$ a \textit{path graph} if $V_G$ is an empty set or $G\simeq\,P_n(u,v)\eqdef \, u\scriptstyle \bullet\,\rule[1.5pt]{22pt}{0.4pt}\,\bullet\,\rule[1.5pt]{22pt}{0.4pt}\,\,\hdashrule[1.5pt]{12pt}{0.4pt}{1pt}\,\rule[1.5pt]{22pt}{0.4pt}\,\bullet \displaystyle v$ and $u\neq v$ (if $u=v$, we call $G$ a \textbf{cycle}). We say that a digraph $D$ is an \textbf{oriented path} if 
$\ov D\simeq\,P_n(u,v)$
and $a\neq b$ (if $a=b$ we call $D$ an \textbf{oriented cycle}). A digraph $D$ is called \textbf{acyclic} if it contains no oriented cycle $\vec P(a,a)\eqdef$
\tikzsetnextfilename{diagcyclegraphn}%
\begin{tikzpicture}[baseline=(n11.base),label distance=-2pt,xscale=0.65, yscale=0.74]
\node[circle, fill=black, inner sep=0pt, minimum size=3.5pt, label={[name=n11]left:$a$}] (n1) at (0  , 0  ) {};
\node[circle, fill=black, inner sep=0pt, minimum size=3.5pt] (n2) at (1  , 0  ) {};
\node (n3) at (2  , 0  ) {$ $};
\node[circle, fill=black, inner sep=0pt, minimum size=3.5pt] (n4) at (5  , 0  ) {};
\node (n5) at (3  , 0  ) {$ $};
\node[circle, fill=black, inner sep=0pt, minimum size=3.5pt] (n6) at (4  , 0  ) {};
\draw[<-,shorten <= 2.50pt, shorten >= 3.50pt] (n1) .. controls (0.8,0.4) and (4.,0.4) .. (n4);
\draw [line width=1.2pt, shorten <= 2.50pt, line cap=round, dash pattern=on 0pt off 5\pgflinewidth, -, shorten <= .50pt, shorten >= -4.50pt] (n3) to  (n5);
\foreach \x/\y in {1/2, 6/4}
        \draw [->, shorten <= 2.50pt, shorten >= 2.50pt] (n\x) to  (n\y);
\draw [->, shorten <= 2.50pt, shorten >= -2.50pt] (n2) to  (n3);
\draw [->, shorten <= -2.50pt, shorten >= 2.50pt] (n5) to  (n6);
\end{tikzpicture}, i.e., induced subdigraph isomorphic to $\vec P(a,a)$. A graph $G=(V_G,E_G)$ is \textbf{connected} if $P(u, v)\subseteq G$ for every $u\neq v\in V_G$. A digraph $D$ (bigraph $\Delta$) is connected if the graph $\ov D$ ($\ov \Delta$) is connected. A connected (di)graph is called a \textbf{tree} if it does not contain any cycle. We call a vertex $v$ of a digraph $D=(V,A)$ a \textit{source} (minimum) if it is not a target of any edge $\alpha\in A$. Analogously, we call $v\in D$ a \textit{sink} (maximum) if it is not a source of any edge. By degree $\deg(v)$ of a vertex $v$ in (di)graph $G$  we mean a number of edges incident with $v$. A (di)graph $G$  is called $2$-regular if $\deg(v)=2$ for every $v\in G$.

Every (bi)graph $\Delta$ is uniquely determined by the non-symmetric Gram matrix $\check G_\Delta\in\MM_n(\ZZ)$
\begin{equation}\label{eg:bigr:nonsymgrammat}
\check G_\Delta\eqdef\begin{bmatrix*}[r]
1+d_{11}^\Delta & d_{12}^\Delta & \cdots & d_{1n}^\Delta\\
0 & 1+d_{22}^\Delta & \cdots & d_{2n}^\Delta\\
\vdots & \vdots & \ddots & \vdots\\
0 & 0 & \hdots & 1+d_{nn}^\Delta
\end{bmatrix*},
\textnormal{ where }
d_{ij}^\Delta:=|\Delta_1^+(i,j)|-|\Delta_1^-(i,j)|
\end{equation}
and the symmetric Gram matrix 
$G_\Delta:=\tfrac{1}{2}(\check G_\Delta+\check G_\Delta^{tr})\in\MM_n(\QQ)$,
see~\cite{simsonCoxeterGramClassificationPositive2013}. 
\smallskip

By a finite poset $I= (\{1,\ldots,n\}, \preceq)$ we mean a partially ordered set $I$, 
with respect to a partial order relation $\preceq$. Every finite poset $I$ is uniquely encoded in the form of its \textit{incidence matrix} 
\begin{equation}\label{eq:pos_inc_mat}
C_I=[c_{ij}]\in\MM_{n}(\ZZ),
\textnormal{ where } c_{ij}=1 \textnormal{ if } i \preceq j \textnormal{ and }c_{ij}=0 \textnormal{ otherwise}, 
\end{equation}
see~\cite{SimZaj_intmms}.
We often use the \textit{Hasse digraph} $\CH(I)$ representation of $I$, where $\CH(I)=(V,A)$ is an acyclic digraph with vertices $V=\{1,\ldots,n\}$ and edges defined as follows: there is an oriented edge (arrow) $i\to j\in A$ if and only if $i\preceq j$ and there is no such a $k\in\{1,\ldots,n\}\setminus \{i,j\}$ that $i\preceq k\preceq j$, see \cite[Section 14.1]{Si92}. We note that $\CH(I)$ encodes $I$\ uniquely. We associate with $I$ its (incidence) bigraph $\Delta_I=(V,A)$, where $V=\{1,\ldots,n\}$ and $i\,\hdashrule[3pt]{18pt}{0.4pt}{1pt}j\in A$ iff $i\preceq j$. We call $I$ \textbf{connected} if the bigraph $\Delta_I$ or, equivalently,
the digraph $\CH(I)$ is connected. Following~\cite{SimZaj_intmms} we associate with $I= (\{1,\ldots,n\}, \preceq)$:
\begin{itemize}
        \item the symmetric Gram matrix $\smash{G_I:=\frac{1}{2}(C_I+C_I^{tr})\in\MM_n(\QQ)}$,\inlineeqno{eq:pos_gram_mat}
        \item the incidence bilinear form
        $b_I\colon \ZZ^n\times\ZZ^n\to\ZZ$, 
        $b_I(x,y)\eqdef\sum_{i \preceq j}x_i y_j=x\cdot C_I\cdot y^{tr}$,\inlineeqno{eq:bilinear_form}
        \item the incidence quadratic form
        $q_I\colon \ZZ^n\to\ZZ$,
        \begin{equation}\label{eq:incidence_form}
        q_I(x)\eqdef b_I(x,x)=\sum_{i \preceq j}x_i x_j=\sum_{i}x_i^2+ \sum_{i \prec j}x_i x_j= x\cdot C_I\cdot x^{tr} =x\cdot G_I\cdot x^{tr},
        \end{equation}
        \item the Coxeter matrix $\smash{\Cox_I:=-C_I\cdot C_I^{-tr}\in\MM_n(\ZZ)}$,
where $\smash{C_I^{-tr}:=(C_I^{tr})^{-1}=(C_I^{-1})^{tr}}$,\inlineeqno{eq:pos_cox_mat}
        \item the Coxeter polynomial $\cox_I(t):=\det (tE-\Cox_I)\in\ZZ[t]$,\inlineeqno{eq:pos_cox_poly}
        \item the Coxeter spectrum $\specc_I:=\{\lambda \in \CCC;\, \cox_I(\lambda)=0\}\subseteq\CCC$, \inlineeqno{eq:pos_cox_specc}

that is, the multiset of all eigenvalues of the Coxeter matrix $\Cox_I\in\MM_n(\ZZ)$ (with multiplicities).
\end{itemize}
We note that $C_I=\cG_{\Delta_I}$ if and only if vertices of the digraph $\CH(I)$ are topologically ordered.

A poset $I$ is called \textit{non-negative} [of corank $\crk_I\geq 0$] if its symmetric Gram matrix \eqref{eq:pos_gram_mat} is positive semi-definite of rank $n-\crk_I\geq 0$, see~\cite{SimZaj_intmms}. By a \textit{positive} [\textit{principal}] poset, we mean a non-negative $I$ of corank $\crk_I=0$ [$\crk_I=1$].
\begin{definition}\label{df:congruences}
Two posets $I$ and $J$ are said to be:
strongly [weakly] Gram $\ZZ$\hyp congruent and denoted by $I\approx_\ZZ J$
[$I\sim_\ZZ J$] if their incident matrices \eqref{eq:pos_inc_mat} [symmetric Gram matrices \eqref{eq:pos_gram_mat}] are 
$\ZZ$\hyp congruent, i.e. 
$C_I\sim_\ZZ C_J$ [$G_I\sim_\ZZ G_J$]. A poset $I$ is strongly [weakly] Gram $\ZZ$\hyp congruent with a bigraph $\Delta'$, i.e., $I\approx_\ZZ \Delta'$
[$I\sim_\ZZ \Delta'$] if $C_I\sim_\ZZ \cG_{\Delta'}$ [$G_I\sim_\ZZ G_{\Delta'}$].
\end{definition}
It is straightforward to check that
\begin{equation}\label{eq:isomorphism:stronggram}
I\simeq J \Rightarrow I\approx_\ZZ J 
\end{equation}
for any two posets $I$ and $J$, that is, isomorphic posets are strongly Gram $\ZZ$\hyp congruent. In this case, the $\ZZ$-congruence is defined by the permutation matrix $B_\sigma\in\Gln$, where $\sigma\colon I\to J$ is a bijection defining $I\simeq J$ isomorphism.

Given a non-negative poset $I$, following~\cite{SimZaj_intmms}
and~\cite{barotQuadraticFormsCombinatorics2019,simsonCoxeterGramClassificationPositive2013,simsonFrameworkCoxeterSpectral2013}, we use the following definitions.
\begin{itemize}
\item The Coxeter number $\bc_I\in\NN$ is the order of $\Cox_I$ \eqref{eq:pos_cox_mat} in the group $\Gln$, that is, such a minimal integer $\bc_I\geq 1$ that $\Cox_I^{\bc_I}=E$. If such a number does not exist, we set $\bc_I\eqdef\infty.$\inlineeqno{eq:pos_cox_num}
\item The reduced Coxeter number $\check \bc_I\in\NN$, that is, such a minimal $\check \bc_I\geq 1$, that for every $1\leq i  \leq n$
\begin{equation}\label{eq:pos_red_cox_num}
e_i\cdot(\Cox_I^{\check\bc_I}-E)\in\Ker q_I,\textnormal{ where }
\Ker q_I\eqdef \{v\in\ZZ^n;\, q_I(v)=0\}.
\end{equation}
\end{itemize}

The basic properties of a strong Gram $\ZZ$-congruence are summarized in the following fact. In particular, we note that Coxeter spectrum \eqref{eq:pos_cox_specc}, Coxeter polynomial \eqref{eq:pos_cox_poly}, Coxeter number \eqref{eq:pos_red_cox_num} and reduced Coxeter number \eqref{eq:pos_red_cox_num} are invariant under strong Gram $\ZZ$-congruence.
\begin{fact}\label{fact:sgc_conseq}
        Assume that $I$ and $J$ are finite partially ordered sets, and $\Delta$ is a bigraph.
        \begin{enumerate}[label={\textnormal{(\makebox[\widthof{d}][c]{\alph*})}}]
                \item\label{fact:sgc_conseq:coxinvariants} $I\approx_\ZZ J \Rightarrow \specc_I=\specc_J,\ \cox_I(t)=\cox_J(t),\ \bc_I=\bc_J$\textnormal{ and }$\check \bc_I=\check \bc_J$
                \item\label{fact:sgc_conseq:weakcongr} $I\approx_\ZZ J \Rightarrow I\sim_\ZZ J$
                \item\label{fact:sgc_conseq:nnegcorank} If poset $I$ [bigraph $\Delta$] is non-negative of corank $r$ and $I\sim_\ZZ J$ [$I\sim_\ZZ \Delta$], then poset $J$ [bigraph~$\Delta$] is non-negative of corank $r$.
        \end{enumerate}
\end{fact}
\begin{proof}
        Apply arguments of~\cite[Lemma~2.1]{simsonCoxeterGramClassificationPositive2013} and \cite[Lemma 3]{SimZaj_intmms}.
\end{proof}

In the case of non-negative posets $I$, the kernel $\Ker q_I$ is a group that admits a 
$(k_1,\ldots, k_r)$-special $\ZZ$-basis in the following sense.

\begin{fact}\label{fact:specialzbasis} Assume that $I=(\{1,\ldots,n\}, \preceq)$ is a connected non-negative poset
of corank $r \geq 1$.%
\begin{enumerate}[label=\normalfont{(\alph*)}]
\item\label{fact:specialzbasis:existance} There exist integers $1 \leq j_1 < \ldots < j_r \leq n$ such that free abelian group $\Ker q_I\subseteq \ZZ^n$ of rank $r \geq 1$ admits a 
$(k_1,\ldots, k_r)$-special $\ZZ$-basis $h^{(k_1)},\ldots, h^{(k_r)}\in\Ker q_I$, that is, $h^{(k_i)}_{k_i} = 1$ and $h^{(k_i)}_{k_j} = 0$ for $1 \leq i,j \leq r$ and $i \neq j$.
\item\label{fact:specialzbasis:subbigraph} $I^{(k_i)}\eqdef I\setminus\{k_i\}$ is a connected non-negative poset
of corank $r - 1\geq 0$.
\item\label{fact:specialzbasis:positsubbigraph} Poset $I^{(k_1,\ldots,k_r)}\eqdef I\setminus\{k_1,\ldots,k_r\}$ is of corank $0$ (i.e., positive) and connected.\inlineeqno{eq:positive_subbigraph}
\end{enumerate}
\end{fact}
\begin{proof}
Since, without loss of generality, one may assume that $C_I=\cG_{\Delta_I}$ (i.e., vertices of the digraph $\CH(I)$ are topologically ordered), apply \cite[Proposition 5.1]{simsonSymbolicAlgorithmsComputing2016a} and \cite[Theorem 2.1]{zajacStructureLoopfreeNonnegative2019}.
\end{proof}
\vspace*{-1ex}
{\newcommand{\mxs}{0.80}
        \tikzexternaldisable
        \begin{longtable}{@{}r@{\,}l@{\,}l@{\quad}r@{\,}l@{}}
                $\AA_n\colon$ & \grapheAn{\mxs}{1} & $\scriptstyle (n\geq 1);$\\[0.2cm]
                $\DD_n\colon$ & \grapheDn{\mxs}{1} & $\scriptstyle (n\geq 1);$ & $\EE_6\colon$ & \grapheEsix{\mxs}{1}\\[0.2cm]
                $\EE_7\colon$ & \grapheEseven{\mxs}{1} &  & $\EE_8\colon$ & \grapheEeight{\mxs}{1}\\
                \caption{Simply-laced Dynkin diagrams}\label{tbl:Dynkin_diagrams}
                \tikzexternalenable
\end{longtable}}%
\vspace*{-1ex}

Following~\cite{SimZaj_intmms,gasiorekAlgorithmicCoxeterSpectral2020,barotQuadraticFormsCombinatorics2019} with any connected non-negative poset $I$ we associate its \textbf{Dynkin type} $\Dyn_I\in \{\AA_n,\DD_n,\EE_6,\EE_7,\EE_8\}$.\pagebreak

\begin{definition}\label{df:Dynkin_type}
Assume that $I$ is a connected non-negative poset  of corank $r\geq 0$.
The Dynkin type $\Dyn_I$ is defined to be the unique simply-laced Dynkin diagram of Table~\ref{tbl:Dynkin_diagrams} viewed as a bigraph
\[
\Dyn_I  \in \{\AA_n,\DD_n,\EE_6,\EE_7,\EE_8\}
\]
such that $\check \Delta_I$ is weakly Gram $\ZZ$-congruent with $\Dyn_I$, 
where
\begin{itemize}
\item $\check \Delta_I\eqdef \Delta_I$ if $r=0$ (i.e., $I$ is positive),
\item $\check \Delta_I\eqdef\Delta_I^{(k_1,\ldots,k_r)}=\Delta_{I\setminus\{k_1,\ldots,k_r\}}\subseteq\Delta_I$~\eqref{eq:positive_subbigraph} if $r>0$.
\end{itemize}
The bigraph $\Dyn_I$ can be obtained by means of the inflation algorithm~\cite[Algorithm 4.6]{zajacPolynomialTimeInflation2020}, in a polynomial time~\cite[Corollary 4.1]{zajacPolynomialTimeInflation2020}.
\end{definition}

In \cite{gasiorekStructureNonnegativePosets2022arxiv}, we described the non-negative posets $I$ of Dynkin type $\Dyn_I=\AA_m$ in terms of their Hasse digraphs as follows. 
\begin{fact}\label{fact:a:hasse}
Assume that $I=(\{1,\ldots n\},\preceq)$ is a finite connected poset. %
\begin{enumerate}[label=\normalfont{(\alph*)}]
\item\label{fact:a:hasse:posit} $I$ is non-negative of Dynkin type $\Dyn_I=\AA_n$ if and only if
        $\CH(I)$ is an oriented path.
\item\label{fact:a:hasse:princ} $I$ is non-negative of Dynkin type $\Dyn_I=\AA_{n-1}$  if and only if $\CH(I)$ is $2$-regular and $I$ has at least two maximal elements.
\item\label{fact:a:hasse:crkbiggeri} If $I$ is non-negative of Dynkin type $\Dyn_I=\AA_{n-\crk_I}$, then $\crk_I\in\{0,1\}$.
\end{enumerate}
\end{fact}

Summing up, every non-negative poset $I$ of Dynkin type $\AA_m$  has its Hasse quiver isomorphic with either an oriented path (and is positive) or an oriented cycle (and is principal). While every oriented path yields a positive poset, this is not the case for oriented cycles: only a cycle with at least two sinks (sources) yields a Hasse digraph of a principal poset of Dynkin type $\AA_m$, see~\cite{gasiorekStructureNonnegativePosets2022arxiv} for a detailed analysis. 

\section{Combinatorial tools}

Inspired by the applications of ($min,max$)-equivalence, introduced by Bondarenko~\cite{bondarenkoMinMaxEquivalence2005} in the study of posets with respect to the non-negativity of Tits quadratic form \cite{bondarenkoClassificationSerialPosets2019,bondarenkoMinMaxEquivalence2008}, we introduce the notion of a ($min,max$)-reflection.
By the \textit{neighbourhood} $N_{I}(a)\subseteq I$ of $a\in I=(I, \preceq)$ in $I$ we mean the set 
\[
N_{I}(a)\eqdef \{b\in I;\, a\preceq b\ \land\ \neg \exists_{c\neq b} a \preceq c \preceq b \} \cup
\{d\in I;\, d\preceq a\ \land\ \neg \exists_{c\neq d} d \preceq c \preceq a \}.
\]
In other words $N_{I}(a)\subseteq I$ is a set of elements that either \textit{covers} $a$ or \textit{are covered} by $a$.

\begin{definition}\label{df:minmaxrefl}
Let $I=(I, \preceq)$ be a finite poset. For a minimal (resp. maximal) element $a\in I$ the ($min,max$)-reflection of $I$ at $a$ is the poset $S_a I\eqdef(I,\preceq_\bullet)$, where $\preceq_\bullet\in I\times I$ is a transitive closure of the $\preceq_\circ$
 relation, defined as follows:

\begin{itemize}
\item for every $b\in N_I(a)$, we have $b \preceq_\circ a$ if and only if $a \preceq b$ (resp. $a \preceq_\circ b$ if and only if $b \preceq a$),
\item for every $c,d\in I\setminus \{a\}$, we have $c \preceq_\circ d$ if and only if $c \preceq d$ .
\end{itemize}
\end{definition}

The following example illustrates the ($min,max$)-reflection operation.
\begin{example}\label{ex:minmaxref}
Consider $5$ element poset 
$J=(\{1,2,3,4,5\}, \{2\preceq 1,\ab 3\preceq 1,\ab 4\preceq 1,\ab 3 \preceq 2,\ab 4\preceq  2,\ab 3 \preceq 5,\ab 4\preceq 5 \})$.
For the maximal element element $5\in J$, with $N_J(5)=\{3,4\}$, we have
\[S_5 J=(\{1,2,3,4,5\}, \{2\preceq_\bullet 1, 3\preceq_\bullet 1, 4\preceq_\bullet 1, 5\preceq_\bullet 1, 3\preceq_\bullet 2, 4\preceq_\bullet 2, 5\preceq_\bullet 2, 5\preceq_\bullet 3, 5\preceq_\bullet 4\}),\]
where the relation $\preceq_\bullet$  is a transitive closure of the relation
 $\preceq_\circ$ defined as 
\[\{2\preceq_\circ 1, 3\preceq_\circ 1, 4\preceq_\circ 1, 3 \preceq_\circ 2, 4\preceq_\circ  2, 5 \preceq_\circ 3, 5\preceq_\circ  4\}\subseteq J\times J.\] 

The ($min,max$)-reflection $J\mapsto S_5J$ has the following interpretation at the Hasse quiver level.
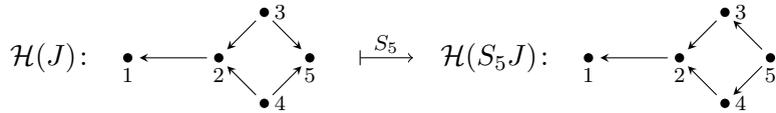
\begin{center}
$\CH(J)\colon$
\tikzsetnextfilename{minmaxExample}
\begin{tikzpicture}[baseline={([yshift=-2.75pt]current bounding box)},
label distance=-2pt,xscale=0.6, yscale=0.6]
\node[circle, fill=black, inner sep=0pt, minimum size=3.5pt, label=below:$\scriptstyle 1$] (n1) at (0  , 1  ) {};
\node[circle, fill=black, inner sep=0pt, minimum size=3.5pt, label=below:$\scriptstyle 2$] (n2) at (2  , 1  ) {};
\node[circle, fill=black, inner sep=0pt, minimum size=3.5pt, label=right:$\scriptstyle 3$] (n3) at (3  , 2  ) {};
\node[circle, fill=black, inner sep=0pt, minimum size=3.5pt, label=right:$\scriptstyle 4$] (n4) at (3  , 0  ) {};
\node[circle, fill=black, inner sep=0pt, minimum size=3.5pt, label=below:$\scriptstyle 5$] (n5) at (4  , 1  ) {};
\foreach \x/\y in {2/1, 3/2, 3/5, 4/2, 4/5}
        \draw [-stealth, shorten <= 2.50pt, shorten >= 2.50pt] (n\x) to  (n\y);
\end{tikzpicture}
\tikzsetnextfilename{minmaxExamplearrii}
\begin{tikzpicture}[baseline={([yshift=-2.75pt]current bounding box)},label distance=-2pt,xscale=0.65, yscale=0.74]
\node (n1) at (0  , 0.40) {$ $};
\node (n3) at (1.50, 0.40) {$ $};
\node (n4) at (1.50, 0  ) {$ $};
\draw [|-to] (n1) to  node[above=-2.0pt, pos=0.5] {{\scriptsize $S_5$}} (n3);
\end{tikzpicture}
$\CH(S_5J)\colon$
\tikzsetnextfilename{minmaxExamplev}
\begin{tikzpicture}[baseline={([yshift=-2.75pt]current bounding box)},
label distance=-2pt,xscale=0.6, yscale=0.6]
\node[circle, fill=black, inner sep=0pt, minimum size=3.5pt, label=below:$\scriptstyle 1$] (n1) at (0  , 1  ) {};
\node[circle, fill=black, inner sep=0pt, minimum size=3.5pt, label=below:$\scriptstyle 2$] (n2) at (2  , 1  ) {};
\node[circle, fill=black, inner sep=0pt, minimum size=3.5pt, label=right:$\scriptstyle 3$] (n3) at (3  , 2  ) {};
\node[circle, fill=black, inner sep=0pt, minimum size=3.5pt, label=right:$\scriptstyle 4$] (n4) at (3  , 0  ) {};
\node[circle, fill=black, inner sep=0pt, minimum size=3.5pt, label=below:$\scriptstyle 5$] (n5) at (4  , 1  ) {};
\foreach \x/\y in {2/1, 3/2, 4/2, 5/3, 5/4}
        \draw [-stealth, shorten <= 2.50pt, shorten >= 2.50pt] (n\x) to  (n\y);
\end{tikzpicture}
\end{center}
\end{example}

As Example~\ref{ex:minmaxref} illustrates, ($min,max$)-reflection $S_aI$ at $a\in I$ can be viewed as a reflection $s_a$ of Hasse quiver $\CH(I)$ \cite[Chapter VII.4]{ASS}. Contrary to the quiver operation and ($min,max$)-equivalence, the $I\mapsto S_aI$ operation does not preserve equivalence of quadratic forms, in general.

\begin{example}\label{ex:minmaxref:def}
Consider $5$ element poset $I=(\{1,2,3,4,5\}, \preceq)$ of Example~\ref{ex:minmaxref}. The quadratic form $q_I\colon \ZZ^5\!\to \ZZ$ \eqref{eq:incidence_form} is indefinite, as $q_I([-4,\!-4,5,7,\!-6])\!=\!-10$,
while the quadratic form $q_{S_5I}\colon \ZZ^5\!\to \ZZ$ is positive, since $S_5I$ is isomorphic with the positive one-peak poset ${}_{1}\DD^*_{4} \diamond \AA_{0}$, see \cite{gasiorekOnepeakPosetsPositive2012}.
\end{example}

\begin{definition}\label{df:junctionfree}
Let $I=(I,\preceq)$ be a finite poset. We call a minimal [maximal] element $a\in I$ \textit{junction free} if and only if for all $b,c\in N_I(a)$ and $d\in I\setminus\{a\}$
we have 
\[(a\preceq b\preceq d)\land (a\preceq c\preceq d) \Rightarrow b=c\qquad [d\preceq b\preceq a)\land (d\preceq c\preceq a) \Rightarrow b=c].
\]
\end{definition}

In the poset $S_5J$ of Example~\ref{ex:minmaxref}, the element $1$ is junction free and $5$ is not, since $2\preceq_\bullet 3 \preceq_\bullet 5$ and $2\preceq_\bullet 4 \preceq_\bullet 5$.

\begin{proposition}\label{prop:minmaxrefl}
Let $I=(\{1,\ldots,n\},\preceq)$ be a finite poset and $a\in I$ be a minimal or maximal element that is junction free in both $I$ and $S_aI$.
\begin{enumerate}[label={\textnormal{(\alph*)}}]
        \item\label{cor:minmaxrefl:diagcomm} The diagram
\begin{equation}\label{cor:minmaxrefl:diag}
\begin{tikzpicture}[baseline={([yshift=-2.75pt]current bounding box)},
label distance=-2pt]
\matrix [matrix of nodes, ampersand replacement=\&, nodes={minimum height=1.5em,minimum width=1.5em,
text depth=0ex,text height=1ex, execute at begin node=$, execute at end node=$}
, column sep={50pt,between origins}, row sep={40pt,between origins}]
{
|(n1)|\ZZ^n\times\ZZ^n  \& |(n2)|\ZZ \\
|(n3)|\ZZ^n\times\ZZ^n  \& \\
};
\draw [-to, shorten <= -.50pt, shorten >= .50pt] (n1) to  node[above, pos=0.5] {{\scriptsize $b_I$}} (n2);
\draw [-to, shorten <= 2.0pt, shorten >= .50pt] (n3) to  node[below right, pos=0.5] {{\scriptsize $b_{S_aI}$}} (n2);
\draw [-to, shorten <= .50pt, shorten >= .50pt] (n1) to  node[right, pos=0.5] {{\scriptsize $\simeq$}} (n3);
\path (n3) to node[left, pos=0.5] {{\scriptsize $s_a\times s_a$}} (n1);
\end{tikzpicture}
\end{equation}
is commutative, where the group isomorphism $\ZZ^n\ni x\xmapsto{s_a} y\in\ZZ^n$, is defined as follows: $y_a\eqdef-x_a$, $y_i\eqdef x_i+x_a$ for $i\in N_I(a)$ and $y_i\eqdef x_i$ otherwise.
\item\label{cor:minmaxrefl:strongcongr} Poset $I$ is strongly Gram $\ZZ$-congruent with poset $S_aI$.
\item\label{cor:minmaxrefl:invol} $I\mapsto S_aI$ operation is involution, that is $S_a(S_aI) = I$.
\end{enumerate}
\end{proposition}
\begin{proof}
\ref{cor:minmaxrefl:diagcomm} Without loss of generality, we may assume that $a=1$ is a minimal element of 
$I=(\{1,\ldots,n\},\preceq)$ and $N_I(a)=\{2,\ldots,k\}$. By assumptions, $\ZZ$-bilinear form $b_I\colon\ZZ^n\times\ZZ^n\to\ZZ$ \eqref{eq:bilinear_form} has a form
{
\allowdisplaybreaks
\begin{align*}
b_I(x,y) &= x_1y_1+\sum_{i=2}^{k}y_i\left(x_1+x_i+ \sum_{j=k+1}^n c_{j,i}x_j \right)+
\sum_{i=k+1}^{n}y_i\left(x_i+\sum_{j\neq i} c_{j,i}x_{j}\right)\\
&= x_{1}\left(\sum_{i=1}^k y_i + \sum_{i=k+1}^n c_{1,i} y_{i} \right) 
+b_{\check I}(x,y),
\end{align*}
}%
where $c_{i,j}\in\{0,1\}$ and $b_{\check I}(x,y)\eqdef b_I(x|_{x_1=0},y|_{y_1}=0)$ is the $\ZZ$-bilinear form associated with the poset $I\setminus\{1\}$. By direct calculations, one easily checks that
{
\allowdisplaybreaks
\begin{align*}
b_I(s_1(x),s_1(y)) &= x_1\sum_{i=k+1}^{n}y_i\left(\sum_{j=2}^k c_{j,i} - c_{1,i} \right)+
y_1\left(\sum_{i=1}^{k} x_i +\sum_{i=k+1}^{n}x_i\sum_{j=2}^{k} c_{i,j} \right)\ 
+b_{\check I}(x,y).
\end{align*}
}%

We claim that $\ZZ$-bilinear form $b_I(s_1(x),s_1(y))$ describes poset $S_1I$, i.e., $b_{S_1I} = b_I(s_1(x),s_1(y))$.
Indeed, 
for any $i>k$ and $j\in\{2,\ldots,k\}$ we have the following.
\begin{enumerate}[label={\textnormal{(\roman*)}}]
\item At most one $c_{j,i}\neq 0,$ since $1\preceq j$ and $1$ is \textit{junction free} in $I$. If that is the case, $c_{1,i}=1$ by transitivity of the relation $\preceq$ and, consequently, we have 
$x_1\sum_{i=k+1}^{n}y_i\left(\sum_{j=2}^k c_{j,i} - c_{1,i} \right)=0$.
\item At most one $c_{i,j}\neq 0,$ since otherwise $1$ would not be \textit{junction free} in $S_1I$.
\end{enumerate}
Summing up, $\ZZ$-bilinear form $b_I(s_1(x),s_1(y))$ describes such a poset $J=(I,\preceq_\bullet)$, that $J\setminus \{1\}$ coincides with $I\setminus \{1\}$, $1$ is a maximal element in $J$ and $b\preceq_\bullet 1$ iff $c\preceq b$ for some $c\in N_I(1)$. Hence $J=S_1I$ by Definition~\ref{df:minmaxrefl}.\smallskip

\ref{cor:minmaxrefl:strongcongr} Follows directly from \ref{cor:minmaxrefl:diagcomm}. In particular, $I \overset{B_a}{\approx_\ZZ} S_aI$, where $B_a=[b^a_{i,j}]\in\MM_n(\ZZ)$
\begin{equation}\label{eq:reflmat}
b^a_{i,j}=
\begin{cases}
-1 & \textnormal{if } i=j= a,\\
\phantom{-}1 & \textnormal{if } i=j\textnormal{ and }i\neq a\textnormal{ or }i\in N_I(a)\textnormal{ and } j=a,\\
\phantom{-}0 & \textnormal{otherwise} \\
\end{cases}
\end{equation}
 is the matrix defining the group isomorphism $s_a\colon x\mapsto x\cdot B_a^{tr}$.
Since $s_a$ is an involution (equivalently: $B_a$ is an involutory matrix),
\ref{cor:minmaxrefl:invol}
follows and the proof is finished. 
\end{proof}

Proposition~\ref{prop:minmaxrefl} describes conditions under which ($min,max$)-reflection defines strong Gram $\ZZ$-congruence. In the following example, we illustrate its applications.

\begin{example}\label{ex:minmaxref:strongcongr}
Consider poset 
$J$ of Example~\ref{ex:minmaxref}, described by the Hasse quiver $\CH(J)$ \eqref{ex:minmaxref:strongcongr:eq}.
We note that every element $a\in\{1,3,4,5\}$ is junction free in $J$, but $a$ is junction free in $S_aJ$ only in the case of $a=1$. Hence, in view of Proposition~\ref{prop:minmaxrefl},
$J$ is strong Gram $\ZZ$-congruent with $S_1J$. 
\begin{equation}\label{ex:minmaxref:strongcongr:eq}
\CH(J)\colon
\tikzsetnextfilename{minmaxExample}
\begin{tikzpicture}[baseline={([yshift=-2.75pt]current bounding box)},
label distance=-2pt,xscale=0.6, yscale=0.6]
\node[circle, fill=black, inner sep=0pt, minimum size=3.5pt, label=below:$\scriptstyle 1$] (n1) at (0  , 1  ) {};
\node[circle, fill=black, inner sep=0pt, minimum size=3.5pt, label=below:$\scriptstyle 2$] (n2) at (2  , 1  ) {};
\node[circle, fill=black, inner sep=0pt, minimum size=3.5pt, label=right:$\scriptstyle 3$] (n3) at (3  , 2  ) {};
\node[circle, fill=black, inner sep=0pt, minimum size=3.5pt, label=right:$\scriptstyle 4$] (n4) at (3  , 0  ) {};
\node[circle, fill=black, inner sep=0pt, minimum size=3.5pt, label=below:$\scriptstyle 5$] (n5) at (4  , 1  ) {};
\foreach \x/\y in {2/1, 3/2, 3/5, 4/2, 4/5}
        \draw [-stealth, shorten <= 2.50pt, shorten >= 2.50pt] (n\x) to  (n\y);
\end{tikzpicture}\quad
\tikzsetnextfilename{minmaxExamplearrix}
\begin{tikzpicture}[baseline={([yshift=-2.75pt]current bounding box)},label distance=-2pt,xscale=0.65, yscale=0.74]
\node (n1) at (0  , 0.40) {$ $};
\node (n2) at (0  , 0  ) {$ $};
\node (n3) at (1.50, 0.40) {$ $};
\draw [|-to] (n1) to  node[above=-2.0pt, pos=0.5] {{\scriptsize $S_1$}} (n3);
\end{tikzpicture}\quad
\CH(S_1J)
\tikzsetnextfilename{minmaxExamplei}
\begin{tikzpicture}[baseline={([yshift=-2.75pt]current bounding box)},
label distance=-2pt,xscale=0.6, yscale=0.6]
\node[circle, fill=black, inner sep=0pt, minimum size=3.5pt, label=below:$\scriptstyle 1$] (n1) at (0  , 1  ) {};
\node[circle, fill=black, inner sep=0pt, minimum size=3.5pt, label=below:$\scriptstyle 2$] (n2) at (2  , 1  ) {};
\node[circle, fill=black, inner sep=0pt, minimum size=3.5pt, label=right:$\scriptstyle 3$] (n3) at (3  , 2  ) {};
\node[circle, fill=black, inner sep=0pt, minimum size=3.5pt, label=right:$\scriptstyle 4$] (n4) at (3  , 0  ) {};
\node[circle, fill=black, inner sep=0pt, minimum size=3.5pt, label=below:$\scriptstyle 5$] (n5) at (4  , 1  ) {};
\foreach \x/\y in {1/2, 3/2, 3/5, 4/2, 4/5}
        \draw [-stealth, shorten <= 2.50pt, shorten >= 2.50pt] (n\x) to  (n\y);
\end{tikzpicture}
\end{equation}
In particular, we have $\ZZ^5 \ni [x_1, x_2, x_3, x_4, x_5] \xmapsto{s_1} [-x_1, x_1+x_2,x_3,x_4,x_5]\in\ZZ^5$,
\begin{align*}
b_J(x,y)&= y_{1} \left(x_{1}\! +\! x_{2}\! +\! x_{3}\! +\! x_{4}\right) + y_{2} \left(x_{2}\! +\! x_{3}\! +\! x_{4}\right) + y_{3} x_{3} + y_{4} x_{4} + y_{5} \left(x_{3}\! +\! x_{4}\! +\! x_{5}\right),\\[0.2cm]
b_J(s_1(x),s_1(y)) &=- y_{1}\! \left(x_{2}\! +\! x_{3}\! +\! x_{4}\right)\! +\! \left(y_{1}\! +\! y_{2}\right)\! \left(x_{1}\! +\! x_{2}\! +\! x_{3}\! +\! x_{4}\right)\!+\! y_{3} x_{3}\! +\! y_{4} x_{4}\!+\! y_{5}\! \left(x_{3}\! +\! x_{4}\! +\! x_{5}\right) \\
&={} y_{1} x_{1} + y_{2} \left(x_{1} + x_{2} + x_{3} + x_{4}\right) + y_{3} x_{3} + y_{4} x_{4}+ y_{5} \left(x_{3} + x_{4} + x_{5}\right)\\
&={}b_{S_1J}(x,y)\\
\intertext{and}
B_1^{tr}\cdot C_J\cdot B_1 &= 
\begin{bsmallmatrix*}[r]-1 & 1 & 0 & 0 & 0\\0 & 1 & 0 & 0 & 0\\0 & 0 & 1 & 0 & 0\\0 & 0 & 0 & 1 & 0\\0 & 0 & 0 & 0 & 1\end{bsmallmatrix*}
\cdot
\begin{bsmallmatrix*}[r]1 & 0 & 0 & 0 & 0\\1 & 1 & 0 & 0 & 0\\1 & 1 & 1 & 0 & 1\\1 & 1 & 0 & 1 & 1\\0 & 0 & 0 & 0 & 1\end{bsmallmatrix*}
\cdot
\begin{bsmallmatrix*}[r]-1 & 0 & 0 & 0 & 0\\1 & 1 & 0 & 0 & 0\\0 & 0 & 1 & 0 & 0\\0 & 0 & 0 & 1 & 0\\0 & 0 & 0 & 0 & 1\end{bsmallmatrix*}
=
\begin{bsmallmatrix*}[r]1 & 1 & 0 & 0 & 0\\0 & 1 & 0 & 0 & 0\\0 & 1 & 1 & 0 & 1\\0 & 1 & 0 & 1 & 1\\0 & 0 & 0 & 0 & 1\end{bsmallmatrix*}
=
C_{S_1J}.
\end{align*}

\end{example}

Assume that $I=(I,\preceq)$, $|I|=n$, is a finite connected principal poset of Dynkin type $\AA_{n-1}$. By Fact~\ref{fact:a:hasse}\ref{fact:a:hasse:princ}, elements $a_i$ of $I$ can be enumerated in such an order $a_1,\ldots,a_n$, that
\begin{itemize}
\item $a_{r_0},\ldots,a_{r_{k-1}}$ denote all elements that are either minimal or maximal in $I$,
\item for every $j\in\{0,\ldots, k-1\}$ and $j'={j+1}\bmod k$ either
\begin{itemize}

\item $a_{r_s}\preceq a_{r_t}$ for every $a_{r_s},a_{r_t}\in \{a_{r_j}, \ldots, a_{r_{j'}} \}$ where $r_s<r_t$, or 
\item $a_{r_s}\succeq a_{r_t}$ for every $a_{r_s},a_{r_t}\in \{a_{r_j}, \ldots, a_{r_{j'}} \}$ where $r_s<r_t$.
\end{itemize}
\end{itemize}
Since $I$ has at least two maximal elements and $\CH(I)$ is a $2$-regular directed graph, it easily follows that $2\leq k=2s$ is an even number. 
\begin{definition}\label{df:anposet:cycleindex}
The \textit{cycle index} $c(I)\in\{\lceil\frac{n}{2}\rceil,\ldots,n\}$ of the Dynkin type $\AA_{|I|-1}$ principal connected poset $I=(\{1,\ldots,n\},\preceq)$ that has exactly $k=2s$ elements that are either minimal or maximal is
\begin{equation}\label{eq:df:cycleindex}
c(I)\!\eqdef\!\max\Big(\sum_{j=0}^{s-1} |\{a_{r_{2j}},\ldots, a_{r_{2j+1}}\}|,\ \sum_{j=1}^{s-1} |\{a_{r_{2j-1}},\ldots, a_{r_{2j}}\}|\!+\!|\{a_{r_k-1},\ldots,a_n,a_{1}\}|\Big)-s.
\end{equation}
\end{definition}
The cycle index $c(I)$ has the following combinatorial interpretation at the Hasse quiver $\CH(I)$ level. 
\begin{remark}\label{rmk:cycle_index_hasse}
Let $I=(\{1,\ldots,n\},\preceq)$ be a finite connected principal poset of Dynkin type $\AA_{n-1}$. 
We may assume that the (planar) digraph $\CH(I)$ is visualized graphically in the circle layout, that is 
\begin{equation}\label{eq:princ:hasse:pic}
\CH(I)\colon
\tikzsetnextfilename{graphcircposprinc}
\begin{tikzpicture}[baseline={([yshift=-2.75pt]current bounding box)},label distance=-2pt,
xscale=0.64, yscale=0.64]
\draw[draw, fill=cyan!10](0,1) circle (7pt);
\draw[draw, fill=cyan!10](5,2) circle (7pt);
\draw[draw, fill=cyan!10](9,1) circle (7pt);
\draw[draw, fill=cyan!10](5,0) circle (7pt);
\node[circle, fill=black, inner sep=0pt, minimum size=3.5pt, label={[yshift=-0.5ex,xshift=-0.6ex]left:{$\scriptstyle a_{r_0}=a_{1}$}}] (n1) at (0  , 1  ) {};
\node[circle, fill=black, inner sep=0pt, minimum size=3.5pt, label={[yshift=0.6ex]above:{$\scriptstyle a_{2}$}}] (n2) at (1  , 2  ) {};
\node (n3) at (2  , 2  ) {$\scriptstyle $};
\node (n4) at (3  , 2  ) {$\scriptstyle $};
\node[circle, fill=black, inner sep=0pt, minimum size=3.5pt, label={[xshift=-0.6ex,yshift=0.6ex]above:{$\scriptstyle a_{r_2\mppmss1}$}}] (n5) at (4  , 2  ) {};
\node[circle, fill=black, inner sep=0pt, minimum size=3.5pt, label={[yshift=0.6ex]above:{$\scriptstyle a_{r_1}$}}] (n6) at (5  , 2  ) {};
\node[circle, fill=black, inner sep=0pt, minimum size=3.5pt, label={[xshift=0.6ex,yshift=0.6ex]above:{$\scriptstyle a_{r_2\mpppss1}$}}] (n7) at (6  , 2  ) {};
\node (n8) at (7  , 2  ) {$\scriptstyle $};
\node (n9) at (8  , 2  ) {$\scriptstyle $};
\node[circle, fill=black, inner sep=0pt, minimum size=3.5pt, label={[xshift=0.6ex]right:{$\scriptstyle a_{r_t}$}}] (n10) at (9  , 1  ) {};
\node[circle, fill=black, inner sep=0pt, minimum size=3.5pt, label={[xshift=0.6ex,yshift=-0.6ex]below:{$\scriptstyle a_{r_t\mpppss1}$}}] (n11) at (8  , 0  ) {};
\node (n12) at (7  , 0  ) {$\scriptstyle $};
\node (n13) at (6  , 0  ) {$\scriptstyle $};
\node[circle, fill=black, inner sep=0pt, minimum size=3.5pt,label={[yshift=-0.6ex]below:{$\scriptstyle a_{n\mppmss 1}$}}] (n14) at (2  , 0  ) {};
\node[circle, fill=black, inner sep=0pt, minimum size=3.5pt, label={[xshift=-0.6ex,yshift=-0.6ex]below:{$\scriptstyle a_{n}$}}] (n15) at (1  , 0  ) {};
\node (n16) at (3  , 0  ) {};
\node (n17) at (4  , 0  ) {};
\node[circle, fill=black, inner sep=0pt, minimum size=3.5pt, label={[yshift=-0.6ex]below:{$\scriptstyle a_{r_{k-1}}$}}] (n18) at (5  , 0  ) {};

\foreach \x/\y in {1/2, 5/6, 6/7, 10/11, 14/15, 15/1}
        \draw [-, shorten <= 2.50pt, shorten >= 2.50pt] (n\x) to  (n\y);
\foreach \x/\y in {2/3, 7/8, 11/12, 18/17}
        \draw [-, shorten <= 2.50pt, shorten >= -2.50pt] (n\x) to  (n\y);
\foreach \x/\y in {3/4, 8/9, 12/13, 17/16}
        \draw [line width=1.2pt, line cap=round, dash pattern=on 0pt off 5\pgflinewidth, -, shorten <= -1.00pt, shorten >= -2.50pt] (n\x) to  (n\y);
\foreach \x/\y in {4/5, 9/10, 13/18, 16/14}
        \draw [-, shorten <= -2.50pt, shorten >= 2.50pt] (n\x) to  (n\y);
\end{tikzpicture},
\end{equation}
every arrow $(a_i,a_{(i+1)\bmod n})$ in \eqref{eq:princ:hasse:pic} is oriented either clockwise or counterclockwise and every subquiver $\{a_{r_j}, \ldots, a_{r_{j'}}\}\subseteq \CH(I)$, where $j'={j+1}\bmod k$, is an oriented chain. The cycle index $c(I)$ equals $\max(l,r)$, where $r$
denotes the number of arrows oriented clockwise, and $l$ denotes the number of arrows oriented counterclockwise in \eqref{eq:princ:hasse:pic}. 
\end{remark}
We show in Theorem~\ref{thm:mainthm:stronggram} that $c(I)$ is an invariant of strong Gram $\ZZ$-congruence. Moreover, it is uniquely determined by the Coxeter polynomial $\cox_I(t)\in\ZZ[t]$.
\section{Proof of the main theorem}

The general idea of the proof of the Theorem~\ref{thm:mainthm:stronggram} is to reduce an arbitrary connected non-negative poset $I$ of Dynkin type $\AA_m$ to a \textit{canonical} one. 

\begin{definition}\label{df:canonpan}
Let $n\geq 4$ be a natural number. By a canonical two peak poset ${}_p\wt \AA_{n}$ we mean any of the finite connected posets defined by the following Hasse quiver %
\begin{equation}\label{hasse:pospprincpan}
\CH({}_p\wt\AA_n)\colon
\graphOnePeakpanPrinc, 
\end{equation}  
where $\frac{n}{2}\leq p\leq n-2$.
\end{definition}

In the following lemma, we sum up  some of the properties of canonical posets ${}_p\wt \AA_{n}$ that are important from the Coxeter spectral analysis point of view. %

\begin{lemma}\label{lemma:posprincpan}
If $I\eqdef {}_p\wt\AA_n=(\{1,\ldots,n\}, \preceq)$ is a canonical two peak poset \eqref{hasse:pospprincpan},
then:
\begin{enumerate}[label={\textnormal{(\alph*)}}]
\item\label{lemma:posprincpan:nnegDyn} $I$ is principal, $\Dyn_I=\AA_{n-1}$ and $\Ker I=\ZZ\cdot\bh_I\subset \ZZ^n$, where $\bh^I=[-1,-1,0,\ldots,0,1,1]$,
\item\label{lemma:posprincpan:coxpol} $\cox_I(t)=t^n- t^p-t^{n-p}+1=(t-1)^2\nu_p\nu_{n-p} \in\ZZ[t]$, where 
$\nu_p\eqdef 1+t+t^2+\cdots +t^{n-1}$,
\item\label{lemma:posprincpan:coxnum} $\bc_I=\infty$ and $\check \bc_I=\lcm(p,n-p)$,
\item\label{lemma:posprincpan:cycleindex} $c({}_p\wt\AA_n)=p$.
\end{enumerate}
\end{lemma}
\begin{proof}
\ref{lemma:posprincpan:nnegDyn} The incidence quadratic form $q_{I}\colon \ZZ^n\to \ZZ$ \eqref{eq:incidence_form} is given by the formula:{\allowdisplaybreaks\begin{align*}
q_{I}(x) =& \sum_{i} x_i^2 + \sum_{\mathclap{\substack{i< j \\ i,j\in\CI_1}}} x_ix_j + 
\sum_{\mathclap{\substack{i< j \\ i,j\in\CI_2}}} x_ix_j
+
x_{n-1}(x_1+x_2)
\\
=& \frac{1}{2} \sum_{i=3}^{n-2} x_i^2 +
\frac{1}{2}\Big(\sum_{i\in\CI_1} x_i\Big)^2 + 
\frac{1}{2}\Big(\sum_{i\in\CI_2} x_i\Big)^2 +
\frac{1}{2}(x_1+x_{n-1})^2 + \frac{1}{2}(x_2+x_{n-1})^2,
\end{align*}}%
where $\CI_1\eqdef\{1,3,\ldots,p,n \}$ and $\CI_2\eqdef\{2,p+1,\ldots,n-2,n \}$. It follows that $q_I(v)\geq 0$ for every $v\in\ZZ^n$ and $\bh^I\in \Ker I$, that is, $I$ is non-negative of corank $\crk_I>0$. Since $\bh^I_{n-1}=1$ and, by \cite[Theorem 5.2]{gasiorekOnepeakPosetsPositive2012}, the subposet $I^{(n-1)}\eqdef I\setminus\{n-1\}={}_p\AA^*_{n-1}$ is positive of Dynkin type $\Dyn_{I^{(n-1)}}=\AA_{n-1}$, we conclude that $\crk_I=1$ ($I$ is principal), $\Dyn_I=\AA_{n-1}$
and $\Ker_I=\bh^I\ZZ\subseteq\ZZ^n$, see Definition~\ref{df:Dynkin_type}.\smallskip

To prove \ref{lemma:posprincpan:coxpol} and \ref{lemma:posprincpan:coxnum} we follow the \textit{line graph}~\cite{beinekeLineGraphsLine2021} technique developed by Jiménez González \cite{jimenezgonzalezCoxeterInvariantsNonnegative2022,%
jimenezgonzalezIncidenceGraphsNonnegative2018,%
jimenezgonzalezGraphTheoreticalFramework2021} in the context of Coxeter spectral classification of Dynkin type $\AA_n$ non-negative edge\hyp bipartite signed graphs (bigraphs). First, we note that elements of the poset 
${}_p\wt\AA_n=(\{1,\ldots,n\}, \preceq)$ are topologically sorted, i.e., $i\preceq j$ implies that $i\leq j$, and the incidence matrix $C_I\in\MM_n(\ZZ)$ is upper triangular. Hence $\check G_{\Delta_I}=C_I$ and we conclude that $\cox_I(t)=\cox_{\Delta_I}(t)$, $\bc_I=\bc_{\Delta_I}$ and $\check \bc_I=\check\bc_{\Delta_I}$.

We recall from~\cite{jimenezgonzalezIncidenceGraphsNonnegative2018} (see also~\cite{jimenezgonzalezGraphTheoreticalFramework2021}) that the \textit{line graph} (known also as \textit{root graph} or \textit{incidence graph}) associated with a finite quiver (digraph) $Q=(Q_0, Q_1,s, t)$ is the signed graph $L(Q)=(Q_1, E)$ whose  set of edges $E$ is defined as follows: two vertices $e_1,e_2\in Q_1$ are connected:
\begin{itemize}
\item by a \textit{positive} edge 
$e_1\hdashrule[3pt]{25pt}{0.4pt}{1pt}e_2\in E$
iff $s(e_1)=s(e_2)$ or $t(e_1)=t(e_2)$,
\item by a \textit{negative} edge 
$e_1\rule[3pt]{22pt}{0.4pt}e_2\in E$
iff $s(e_1)=t(e_2)$ or $t(e_1)=s(e_2)$.
\end{itemize}
It is straightforward to check that $\Delta_I=L(Q_{I})$, where 
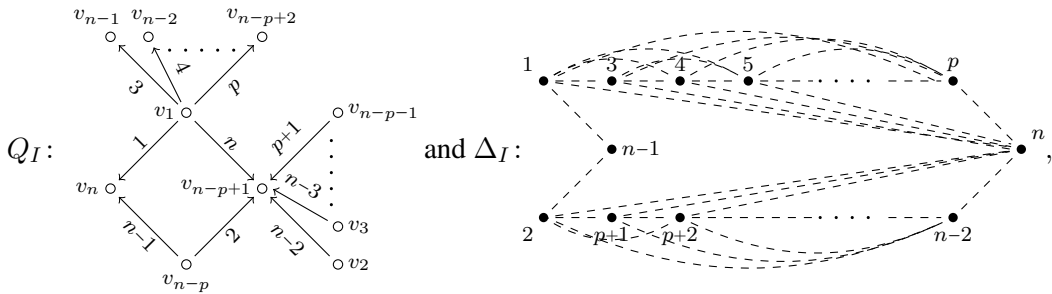
\begin{equation*}%
Q_{I}\colon
\tikzsetnextfilename{linegraphpan}
\begin{tikzpicture}[baseline={([yshift=-2.75pt]current bounding box)},
label distance=-2pt,xscale=0.5, yscale=0.5]
\node[circle, draw, inner sep=0pt, minimum size=3.5pt, label=left:$\scriptstyle v_{1}$] (n1) at (2  , 4  ) {};
\node[circle, draw, inner sep=0pt, minimum size=3.5pt, label=below:$\scriptstyle v_{n\mppms p}$] (n2) at (2  , 0  ) {};
\node[circle, draw, inner sep=0pt, minimum size=3.5pt, label=left:$\scriptstyle v_{n}$] (n3) at (0  , 2  ) {};
\node[circle, draw, inner sep=0pt, minimum size=3.5pt, label=left:$\scriptstyle v_{n\mppms p\mppps 1}$] (n4) at (4  , 2  ) {};
\node[circle, draw, inner sep=0pt, minimum size=3.5pt, label={[xshift=-1ex]above:$\scriptstyle v_{n\mppms 1}$}] (n5) at (0  , 6  ) {};
\node[circle, draw, inner sep=0pt, minimum size=3.5pt, label={[xshift=0.5ex,yshift=0ex]above:$\scriptstyle v_{n\mppms 2}$}] (n6) at (1  , 6  ) {};
\node[circle, draw, inner sep=0pt, minimum size=3.5pt, label=above:$\scriptstyle v_{n\mppms p\mppps 2}$] (n7) at (4  , 6  ) {};
\node[circle, draw, inner sep=0pt, minimum size=3.5pt, label=right:$\scriptstyle v_{2}$] (n8) at (6  , 0  ) {};
\node[circle, draw, inner sep=0pt, minimum size=3.5pt, label=right:$\scriptstyle v_{3}$] (n9) at (6  , 1  ) {};
\node[circle, draw, inner sep=0pt, minimum size=3.5pt, label=right:$\scriptstyle v_{n\mppms p\mppms 1}$] (n10) at (6  , 4  ) {};
\draw [-to, shorten <= 2.50pt, shorten >= 2.50pt] (n2) to  node[below=-2.0pt, pos=0.5, sloped] {{\scriptsize $n\mmm 1$}} (n3);
\draw [-to, shorten <= 2.50pt, shorten >= 2.50pt] (n2) to  node[below=-2.0pt, pos=0.5, sloped] {{\scriptsize $2$}} (n4);
\draw [-to, shorten <= 2.50pt, shorten >= 2.50pt] (n1) to  node[above=-2.0pt, pos=0.5, sloped] {{\scriptsize $1$}} (n3);
\draw [-to, shorten <= 2.50pt, shorten >= 2.50pt] (n1) to  node[above=-2.0pt, pos=0.5, sloped] {{\scriptsize $n$}} (n4);
\draw [-to, shorten <= 2.50pt, shorten >= 2.50pt] (n1) to  node[below, pos=0.5, sloped] {{\scriptsize $3$}} (n5);
\draw [-to, shorten <= 2.50pt, shorten >= 2.50pt] (n1) to  node[above, pos=0.5, sloped] {{\scriptsize $4$}} (n6);
\draw [-to, shorten <= 2.50pt, shorten >= 2.50pt] (n1) to  node[below, pos=0.5, sloped] {{\scriptsize $p$}} (n7);
\draw [line width=1.2pt, line cap=round, dash pattern=on 0pt off 5\pgflinewidth, -, shorten <= 6.00pt, shorten >= 8.00pt] ([yshift=-10]n6.east) to  ([yshift=-10]n7.west);
\draw [-to, shorten <= 2.50pt, shorten >= 2.50pt] (n8) to  node[below, pos=0.5, sloped] {{\scriptsize $n\mmm 2$}} (n4);
\draw [-to, shorten <= 2.50pt, shorten >= 2.50pt] (n9) to  node[above, pos=0.6, sloped] {{\scriptsize $n\mmm 3$}} ([yshift=2]n4.east);
\draw [-to, shorten <= 2.50pt, shorten >= 2.50pt] (n10) to  node[above, pos=0.5, sloped] {{\scriptsize $p\ppp 1$}} (n4);
\draw [line width=1.2pt, line cap=round, dash pattern=on 0pt off 5\pgflinewidth, -, shorten <= 8.50pt, shorten >= 4.00pt] ([xshift=-5]n10.south) to  ([xshift=-5]n9.north);
\end{tikzpicture}
\!\textnormal{and}\ \Delta_I\colon \!\!\!\! 
\tikzsetnextfilename{bigraphpan}
\begin{tikzpicture}[baseline={([yshift=-2.75pt]current bounding box)},
label distance=-2pt,xscale=0.9, yscale=0.9]
\node[circle, fill=black, inner sep=0pt, minimum size=3.5pt, label=above left:$\scriptstyle 1$] (n1) at (0  , 2  ) {};
\node[circle, fill=black, inner sep=0pt, minimum size=3.5pt, label=above:$\scriptstyle 3$] (n2) at (1  , 2  ) {};
\node[circle, fill=black, inner sep=0pt, minimum size=3.5pt, label=above:$\scriptstyle 4$] (n3) at (2  , 2  ) {};
\node[circle, fill=black, inner sep=0pt, minimum size=3.5pt, label=above:$\scriptstyle 5$] (n4) at (3  , 2  ) {};
\node (n5) at (4  , 2  ) {$ $};
\node (n6) at (5  , 2  ) {$ $};
\node[circle, fill=black, inner sep=0pt, minimum size=3.5pt, label=above:$\scriptstyle p$] (n7) at (6  , 2  ) {};
\node[circle, fill=black, inner sep=0pt, minimum size=3.5pt, label=above right:$\scriptstyle n$] (n8) at (7  , 1  ) {};
\node[circle, fill=black, inner sep=0pt, minimum size=3.5pt, label=below left:$\scriptstyle 2$] (n9) at (0  , 0  ) {};
\node[circle, fill=black, inner sep=0pt, minimum size=3.5pt, label=below:$\scriptstyle p\mppps 1$] (n10) at (1  , 0  ) {};
\node[circle, fill=black, inner sep=0pt, minimum size=3.5pt, label=below:$\scriptstyle p\mppps 2$] (n11) at (2  , 0  ) {};
\node (n12) at (4  , 0  ) {$ $};
\node (n13) at (5  , 0  ) {$ $};
\node[circle, fill=black, inner sep=0pt, minimum size=3.5pt, label=below:$\scriptstyle n\mppms 2$] (n14) at (6  , 0  ) {};
\node[circle, fill=black, inner sep=0pt, minimum size=3.5pt, label=right:$\scriptstyle n\mppms 1$] (n15) at (1  , 1  ) {};
\foreach \x/\y in {1/2, 1/15, 2/3, 3/4, 7/8, 9/8, 9/10, 9/15, 10/8, 10/11, 11/8, 14/8}
        \draw [dashed, -, shorten <= 2.50pt, shorten >= 2.50pt] (n\x) to  (n\y);
\foreach \x/\y in {1/3, 1/4, 2/4, 3/7, 4/7}
        \draw [bend left, dashed, -, shorten <= 2.50pt, shorten >= 2.50pt] (n\x) to  (n\y);
\foreach \x/\y in {4/5, 11/12}
        \draw [dashed, -, shorten <= 2.50pt, shorten >= -2.50pt] (n\x) to  (n\y);
\foreach \x/\y in {5/6, 12/13}
        \draw [line width=1.2pt, line cap=round, dash pattern=on 0pt off 5\pgflinewidth, -, shorten <= -2.50pt, shorten >= -2.50pt] (n\x) to  (n\y);
\foreach \x/\y in {6/7, 13/14}
        \draw [dashed, -, shorten <= -2.50pt, shorten >= 2.50pt] (n\x) to  (n\y);
\draw [bend left=25.0, dashed, -, shorten <= 1.40pt, shorten >= -2.80pt] ([xshift=3]n1.north east) to  ([xshift=-9]n7.north west);
\draw [bend left=25.0, dashed, -, shorten <= 0.90pt, shorten >= 2.50pt] ([xshift=3.3]n2.north east) to  (n7);
\foreach \x/\y in {1/8, 3/8}
        \draw [dashed, -, shorten <= 2.50pt, shorten >= 2.50pt] (n\x) to  ([yshift=-1]n\y.north west);
\draw [dashed, -, shorten <= 2.50pt, shorten >= 2.50pt] (n2) to  ([yshift=-2]n8.north west);
\draw [dashed, -, shorten <= 2.50pt, shorten >= 2.50pt] (n4) to  (n8.north west);
\draw [bend right, dashed, -, shorten <= 2.50pt, shorten >= 2.50pt] (n9) to  (n11);
\draw [bend right=25.0, dashed, -, shorten <= 0.90pt, shorten >= 2.50pt] ([xshift=-3.3]n9.north east) to  (n14);
\foreach \x/\y in {10/14, 11/14}
        \draw [bend right=25.0, dashed, -, shorten <= 2.50pt, shorten >= 2.50pt] (n\x) to  (n\y);
\end{tikzpicture}\!\!,
\end{equation*}
that is, $\Delta_I$ is a line graph of $Q_I$.
Moreover, quiver $Q_I$ contains exactly two \textit{minimally decreasing walks} in the sense of~\cite{jimenezgonzalezCoxeterInvariantsNonnegative2022}:
\begin{itemize}
\item $v_1\mapsto v_2 \mapsto v_3\mapsto\cdots\mapsto v_{n-p-1}\mapsto v_{n-p}\mapsto v_{1}$
\item $v_{n-p-1}\mapsto v_{n-p+2} \mapsto\cdots\mapsto v_{n-1}\mapsto v_{n}\mapsto v_{n-p+1}$
\end{itemize}
of length $n-p$ and $p$, respectively. Hence, \cite[Theorem 6.3]{jimenezgonzalezCoxeterInvariantsNonnegative2022} and \cite[Corollary 6.4]{jimenezgonzalezCoxeterInvariantsNonnegative2022} yield
\begin{itemize}
\item $\cox_{\Phi}(t)=(t-1)^{\crk_{\Delta_I}-1}(t^{n-p}-1)(t^p-1)=t^n- t^p-t^{n-p}+1\in\ZZ[t]$ and
\item $c_\Phi=\infty$ and $\check c_\Phi=\lcm(p,n-p)$,
\end{itemize}
where $\Phi\eqdef-\cG_{\Delta_I}^{-1}\cdot\cG_{\Delta_I}^{tr}$. Since $\Phi=\Cox_{\Delta_I}^{-1}= \Cox_I^{-1}$,
we conclude that $\bc_I=\bc_\Phi=\infty$, $\check \bc_I=\check \bc_\Phi=\lcm(p,n-p)$
and, in view of \cite[Lemma 2.8]{simsonMeshGeometriesRoot2011}, $\cox_I(t)=t^n\cox_\Phi(\frac{1}{t})=t^n- t^p-t^{n-p}+1\in\ZZ[t]$.\medskip

\ref{lemma:posprincpan:cycleindex} In order to apply Definition~\ref{df:anposet:cycleindex} we enumerate the elements $\{1,\ldots,n\}$ of ${}_p\wt\AA_n$ poset as follows.
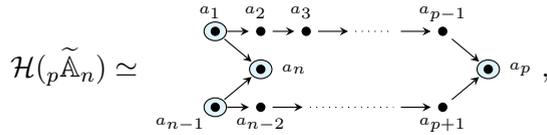
\begin{equation*}%
\CH({}_p\wt\AA_n)\simeq
\tikzsetnextfilename{posetprincpancyclenum}
\begin{tikzpicture}[baseline={([yshift=-2.75pt]current bounding box)},
label distance=-2pt,xscale=0.6, yscale=0.5]
\draw[draw, fill=cyan!10](0,2) circle (7pt);
\draw[draw, fill=cyan!10](0,0) circle (7pt);
\draw[draw, fill=cyan!10](1,1) circle (7pt);
\draw[draw, fill=cyan!10](6,1) circle (7pt);
\node[circle, fill=black, inner sep=0pt, minimum size=3.5pt, label=above:$\scriptscriptstyle a_{1\phantom{p}}$] (n1) at (0  , 2  ) {};
\node[circle, fill=black, inner sep=0pt, minimum size=3.5pt, label=below left:$\scriptscriptstyle a_{n-1}$] (n2) at (0  , 0  ) {};
\node[circle, fill=black, inner sep=0pt, minimum size=3.5pt, label=above:$\scriptscriptstyle a_{2\phantom{p}}$] (n3) at (1  , 2  ) {};
\node[circle, fill=black, inner sep=0pt, minimum size=3.5pt, label=above:$\scriptscriptstyle a_{3\phantom{p}}$] (n4) at (2  , 2  ) {};
\node (n5) at (3  , 2  ) {$\scriptscriptstyle $};
\node (n6) at (4  , 2  ) {$\scriptscriptstyle $};
\node[circle, fill=black, inner sep=0pt, minimum size=3.5pt, label=right:$\scriptscriptstyle \phantom{a}a_{n}$] (n7) at (1  , 1  ) {};
\node[circle, fill=black, inner sep=0pt, minimum size=3.5pt, label=right:$\scriptscriptstyle \phantom{a}a_{p}$] (n8) at (6  , 1  ) {};
\node[circle, fill=black, inner sep=0pt, minimum size=3.5pt, label=above:$\scriptscriptstyle a_{p-1}$] (n9) at (5  , 2  ) {};
\node[circle, fill=black, inner sep=0pt, minimum size=3.5pt, label=below:$\scriptscriptstyle a_{n-2}$] (n10) at (1  , 0  ) {};
\node (n11) at (2  , 0  ) {$\scriptscriptstyle $};
\node (n12) at (4  , 0  ) {$\scriptscriptstyle $};
\node[circle, fill=black, inner sep=0pt, minimum size=3.5pt, label=below:$\scriptscriptstyle a_{p+1}$] (n13) at (5  , 0  ) {};
\foreach \x/\y in {1/3, 1/7, 2/7, 2/10, 3/4, 9/8, 13/8}
        \draw [-stealth, shorten <= 2.50pt, shorten >= 2.50pt] (n\x) to  (n\y);
\foreach \x/\y in {4/5, 10/11}
        \draw [-stealth, shorten <= 2.50pt, shorten >= -2.50pt] (n\x) to  (n\y);
\foreach \x/\y in {5/6, 11/12}
        \draw [dotted, -, shorten <= -2.50pt, shorten >= -2.50pt] (n\x) to  (n\y);
\foreach \x/\y in {6/9, 12/13}
        \draw [-stealth, shorten <= -2.50pt, shorten >= 2.50pt] (n\x) to  (n\y);
\end{tikzpicture},
\end{equation*}
We note that:
\begin{itemize}
\item $a_{r_0}=a_{1},a_{r_1}=a_{p}, a_{r_2}=a_{n-1}, a_{r_3}=a_{n}$ denote all $k=4=2s$ elements that are either minimal or maximal in ${}_p\wt\AA_n$,
\item for every $j\in\{0,\ldots, k-1\}$ and $j'={j}\bmod k+1$, either:
\begin{itemize}
\item $u\preceq v$ for every  $u<v\in \{a_{r_j}, \ldots, a_{r_{j'}} \}$   or 
\item 
$u\succeq v$  for every $u<v\in \{a_{r_j}, \ldots, a_{r_{j'}} \}$. %
\end{itemize}
\end{itemize}
Since $p\geq n-p$, the cycle index \eqref{eq:df:cycleindex} equals 
\begin{align*}
c({}_p\wt\AA_n)&=\max\left(|\{a_1,\ldots,a_p\}| + |\{a_{n-1},a_n\}|,
|\{a_p,\ldots,a_{n-1}|+|\{a_{n},a_{1}\}|\right)-s\\
&=\max(p+2, n-p+2) - 2 = p,
\end{align*}%
and the proof is finished.
\end{proof}

\begin{remark}
We note that Lemma~\ref{lemma:posprincpan}\ref{lemma:posprincpan:nnegDyn} follows by the results of Simson~\cite[Proposition 2.12]{simsonIncidenceCoalgebrasIntervally2009}
and Boldt~\cite[Proposition 3.3]{boldtMethodsDetermineCoxeter1995}. In our proof, we use the line graph technique to determine not only Coxeter polynomial~\eqref{eq:pos_cox_poly} but also Coxeter number \eqref{eq:pos_cox_num} and reduced Coxeter number \eqref{eq:pos_red_cox_num}.
\end{remark}

Now, we have all the necessary tools to prove one of the paper's main results.

\begin{proofof}{Proof of Theorem 1.1}
Part \ref{thm:mainthm:stronggram:dichotomy} of the theorem follows from Fact~\ref{fact:a:hasse}\ref{fact:a:hasse:crkbiggeri}, see \cite{gasiorekStructureNonnegativePosets2022arxiv} for details.\smallskip

In view of Fact~\ref{fact:sgc_conseq}\ref{fact:sgc_conseq:coxinvariants}, \cite[Theorem 5.2]{gasiorekOnepeakPosetsPositive2012} and Lemma~\ref{lemma:posprincpan}, to prove parts \ref{thm:mainthm:stronggram:posit} and \ref{thm:mainthm:stronggram:princ} it suffices to show that \ref{thm:mainthm:stronggram:posit:opcongr} and \ref{thm:mainthm:stronggram:princ:pancongr} hold true.\smallskip

\ref{thm:mainthm:stronggram:posit:opcongr} Assume that $I$ is a positive connected poset of the Dynkin type $\AA_n$. It follows from Fact~\ref{fact:a:hasse}, that $\ov \CH(I)\simeq\,P_n(1,n)\eqdef \, 1\scriptstyle \bullet\,\rule[1.5pt]{22pt}{0.4pt}\,\bullet\,\rule[1.5pt]{22pt}{0.4pt}\,\,\hdashrule[1.5pt]{12pt}{0.4pt}{1pt}\,\rule[1.5pt]{22pt}{0.4pt}\,\bullet \displaystyle n$, hence, by \eqref{eq:isomorphism:stronggram}, without loos of generality we may assume that the elements of $I=(\{1,\ldots,n\},\preceq)$ are enumerated in such a way, that $N_I(1)=\{2\}$, $N_I(n)=\{n-1\}$ and $N_I(i)=\{i-1,i+1\}$, for every $i\in \{2, \ldots, n-1 \}$.

We have two possibilities: either (i) $1\preceq 2$ or (ii) $2\preceq 1$. First, we assume that $1\preceq 2$, i.e., $\CH(I)$ contains the oriented chain $\vec P(1,k)$, where $1<k\leq n$, as an induced subquiver.
\begin{equation*}
\CH(I)\colon
\tikzsetnextfilename{prfgraphanpositi}
\begin{tikzpicture}[baseline={([yshift=1.75pt]current bounding box)},label distance=-2pt,
xscale=0.64, yscale=1.94]
\node[circle, fill=black, inner sep=0pt, minimum size=3.5pt, label=below:$\scriptstyle 1$] (n1) at (0  , 0  ) {};
\node[circle, fill=black, inner sep=0pt, minimum size=3.5pt, label=below:$\scriptstyle 2$] (n2) at (1  , 0  ) {};
\node (n3) at (2  , 0  ) {$\scriptstyle $};
\node (n4) at (3  , 0  ) {$\scriptstyle $};
\node[circle, fill=black, inner sep=0pt, minimum size=3.5pt, label=below:$\scriptstyle k$] (n5) at (4  , 0  ) {};
\node (n6) at (6  , 0  ) {$\scriptstyle $};
\node[circle, fill=black, inner sep=0pt, minimum size=3.5pt, label=below:$\scriptstyle n\mppms 1$] (n7) at (7  , 0  ) {};
\node[circle, fill=black, inner sep=0pt, minimum size=3.5pt, label=below:$\scriptstyle \phantom{1}n\phantom{1}$] (n8) at (8  , 0  ) {};
\node (n9) at (5  , 0  ) {$ $};
\draw [-stealth, shorten <= 2.50pt, shorten >= 2.50pt] (n1) to  (n2);
\draw [-stealth, shorten <= 2.50pt, shorten >= -2.50pt] (n2) to  (n3);
\foreach \x/\y in {3/4, 9/6}
        \draw [line width=1.2pt, line cap=round, dash pattern=on 0pt off 5\pgflinewidth, -, shorten <= -2.50pt, shorten >= -2.50pt] (n\x) to  (n\y);
\draw [-stealth, shorten <= -2.50pt, shorten >= 2.50pt] (n4) to  (n5);
\draw [-, shorten <= -2.50pt, shorten >= 2.50pt] (n6) to  (n7);
\draw [-, shorten <= 2.50pt, shorten >= 2.50pt] (n7) to  (n8);
\draw [stealth-, shorten <= 2.50pt, shorten >= -2.50pt] (n5) to  (n9);
\end{tikzpicture}
\end{equation*}
If $k=n$ then $I={}_0 \AA_{n-1}^*$ \eqref{hasse:pospositzeroanmi} and we are done. Otherwise, we set $I_1\eqdef I$ and define poset $I_2$ as the result of composition of $k$ ($min,max$)-reflections $S_k,S_{k-1},\ldots,S_1$
\begin{equation}\label{eq:minmax:iter}
I_2=I[k]\eqdef S_1\cdots S_{k-1}S_kI.
\end{equation}
We note that $\CH(I_2)$ contains the oriented chain $\vec P(1,k')$, where $k<k'\leq n$ and $k'\geq k+1$ is minimal in $I_2$. If $k'\neq n$, we repeat the procedure and define $I_3\eqdef I[k']$ \eqref{eq:minmax:iter}. It is easy to see that this procedure has to end in at most $r\leq n-k$ steps at $I_r\simeq{}_0 \AA_{n-1}^*$.
\begin{equation}\label{hasse:pospositzeroanmi}
\CH({}_0\AA_{n-1}^*)=
\graphOnePeakpan
\end{equation}
Hence, in view of Proposition~\ref{prop:minmaxrefl}\ref{cor:minmaxrefl:strongcongr}, we conclude that $I\approx_\ZZ {}_0 \AA_{n-1}^*$. To finish this part of the proof, we note that in the case (ii), it suffices to apply the arguments of (i) to the $I_0\eqdef S_1 I$ poset.\smallskip

\ref{thm:mainthm:stronggram:princ:pancongr} We proceed analogously as in the previous case. Without loss of generality, by \eqref{eq:isomorphism:stronggram} and Fact~\ref{fact:a:hasse}\ref{fact:a:hasse:princ}, we may assume that the elements of $I=(\{1,\ldots,n\},\preceq)$ are enumerated in such a way, that $1$ is a minimal element in $I$, $N_I(1) = \{2,n\}$, $N_I(n)=\{1,n-1\}$ and $N_I(i)=\{i-1,i+ 1\}$ for $i\in\{2,\ldots,n-1\}$. Given a minimal or maximal element $i\in I$, we define:
\begin{equation*}
shift_I(i)\eqdef
\begin{cases}
\min_{k\in\NN} \{k ;\ i\preceq k\textnormal{ and }i\not\preceq (i+k)\bmod n+1\},& \textnormal{if $i$ is a minimal element in $I$},\\
\min_{k\in\NN} \{k;\ k\preceq i\textnormal{ and } (i+k)\bmod n+1\not\preceq i\},& \textnormal{if $i$ is a maximal element in $I$},
\end{cases}
\end{equation*}
and $nxt_I(i)\eqdef (i+shitf_I(i))\bmod n + 1\in I$.
It is straightforward to check that $nxt_I(i)\in I$ is a minimal (maximal) element in $I$ if $i\in I$ is a maximal (minimal) one. Furthermore, by $S_a^bI$, where $a < b < n$, we denote a composition of $b-a-1$ ($min,max$)-reflections $S_a^bI\eqdef S_{a+1}\cdots S_{b-1}S_bI$ and by $S_i^{nxt}I$ we mean $ S_i^rI$ with $r\eqdef nxt_I(i)$. 

 Our aim is to show that $I\approx_\ZZ {}_p\wt \AA_{n}$ \eqref{hasse:pospprincpan}, where $p$ is a non-zero integer. We do it by successively applying ($min,max$)-reflections to reduce $I$ to ${}_p\wt \AA_{n}$. We set $I_1\eqdef I$ and proceed as follows.%
\begin{enumerate}[label={\textbf{Step \arabic*${}^\circ$}},leftmargin=1.5cm]
\item\label{prf:thm:mainthm:stronggram:i} If $2\in I_1$ is a maximal element in $I_1$ we set $I_2\eqdef I_1$. Otherwise, we set $I_2\eqdef S_1^{nxt}I_1$.
\item\label{prf:thm:mainthm:stronggram:ii} If $3\in I_2$ is a minimal element in $I_2$ we set $I_3\eqdef I_2$; otherwise: $I_3\eqdef S_2^{nxt}I_2$. We set $k\eqdef 2.$
\item\label{prf:thm:mainthm:stronggram:iii} We set $k\eqdef k +1$, $r_k\eqdef nxt_{I_k}(3)$ and $t_k\eqdef nxt_{I_k}(r_k)$.
\item\label{prf:thm:mainthm:stronggram:iv} If $t_{k}=1$, then $I_k={}_p\wt \AA_{n}$ and we are done. Otherwise:
\item  we set $I_{k+1}\eqdef S_r^tI_k$ and proceed to \ref{prf:thm:mainthm:stronggram:iii}.

 \end{enumerate}
We claim that this procedure finishes in a finite number of steps. Indeed, given $t_k\eqdef nxt_{I_k}(r_k)\neq 1$ in \ref{prf:thm:mainthm:stronggram:iii}, after applying $I_{k+1}\eqdef S_r^tI_k$ in \ref{prf:thm:mainthm:stronggram:iv}, we have $r_{k+1}=nxt_{I_{k+1}}(3)= t_k >r_k>3$. Since $I$ is a finite poset and $r_k$ is a strictly monotonically increasing sequence, we conclude that $t_{k}\eqdef nxt_{I_{k}}(r_{k}) = 1$ for some $k\geq 2$.

To finish the proof, we note that every poset $I_k$ has at least two minimal and at least two maximal elements. It follows that every minimal (maximal) element is junction-free in the sense of Definition~\ref{df:junctionfree}. Hence, by Proposition~\ref{prop:minmaxrefl}\ref{cor:minmaxrefl:strongcongr}, we obtain that $I\approx_\ZZ {}_p\wt \AA_{n}$ for some $p$. To prove that $p$ coincides with the cycle index $c(I)$ of $I$ \eqref{eq:df:cycleindex}, we recall Remark~\ref{rmk:cycle_index_hasse}: the cycle index can be interpreted in terms of clockwise and counterclockwise edges (arrows) in the Hasse digraph $\CH(I)$. It is easy to see that the ($min,max$)-reflection does not change these numbers, i.e., $c(I)=c(S_a(I)),$ since at the Hasse digraph level, it substitutes one clockwise edge by a counterclockwise one, and one counterclockwise by a clockwise one.
We conclude that $c(I)=c({}_p\wt \AA_{n})=p$, see Lemma~\ref{lemma:posprincpan}\ref{lemma:posprincpan:cycleindex}.
\end{proofof}

In the following corollaries, we show that connected non-negative posets $I$ of Dynkin type $\AA_{m}$ are determined uniquely, up to strong Gram $\ZZ$-congruence, by the Coxeter spectrum $\specc_I\subseteq\CCC$. In the case of positive (i.e., $\crk_I=0$) posets, we have a stronger result: in this case strong Gram $\ZZ$-congruence coincides with weak Gram $\ZZ$-congruence (see Definition~\ref{df:congruences}).
\begin{corollary}\label{cor:posit_eqv}
If $I$ and $J$ are finite connected positive posets of Dynkin type $\AA_{m}$, then: 
\begin{enumerate}[label=\normalfont{(\makebox[\widthof{$c'$}][c]{\alph*})}]
\item\label{cor:posit_eqv:weakcongr} $I\sim_\ZZ J$,
  \item\label{cor:posit_eqv:strongcongr} $I\approx_\ZZ J$,
  \item\label{cor:posit_eqv:specc} $\specc_I=\specc_J$ [$\cox_I(t)=\cox_J(t)$].
\end{enumerate}
\end{corollary}
\begin{proof}
In view of Theorem~\ref{thm:mainthm:stronggram}\ref{thm:mainthm:stronggram:posit:opcongr}, for every two finite connected positive posets $I$ and $J$ of Dynkin type $\AA_{m}$ we have $I\approx_\ZZ{}_0 \AA_{n-1}\approx_\ZZ J$, hence $I\approx_\ZZ{} J$ by transitivity of strong Gram $\ZZ$-congruence. This, by Fact~\ref{fact:sgc_conseq}, implies that $I\sim_\ZZ J$, $\specc_I=\specc_J$ and $\cox_I(t)=\cox_J(t)$.
\end{proof}

\begin{corollary}\label{cor:princ_eqv}
If  $I$ and $J$ are finite connected principal posets of Dynkin type $\AA_{m}$, the following conditions are equivalent:
\begin{enumerate}[label=\normalfont{(\makebox[\widthof{$c'$}][c]{\alph*})}]
  \item\label{cor:princ_eqv:strongcongr} $I\approx_\ZZ J$,
  \item\label{cor:princ_eqv:specc} $\specc_I=\specc_J$ [$\cox_I(t)=\cox_J(t)$],
  \item\label{cor:princ_eqv:cyclenum} $c(I) = c(J).$ 
\end{enumerate}
\end{corollary}
\begin{proof}
Assume that $I$ and $J$ are finite connected principal posets of Dynkin type $\AA_{m}$. Implication 
\ref{cor:princ_eqv:strongcongr} $\Rightarrow$ \ref{cor:princ_eqv:specc}
follows by Fact~\ref{fact:sgc_conseq}\ref{fact:sgc_conseq:coxinvariants} and
\ref{cor:princ_eqv:specc} $\Rightarrow$ \ref{cor:princ_eqv:cyclenum}
 is a consequence of Theorem~\ref{thm:mainthm:stronggram}\ref{thm:mainthm:stronggram:princ}.

\ref{cor:princ_eqv:cyclenum} $\Rightarrow$ \ref{cor:princ_eqv:strongcongr} %
By Theorem~\ref{thm:mainthm:stronggram}\ref{thm:mainthm:stronggram:princ:pancongr}, we know that $I\approx_\ZZ {}_p\wt\AA_n \approx_\ZZ J$, where $p= c(I)=c(J)$, and we conclude that $I \approx_\ZZ J$ by transitivity of strong Gram $\ZZ$-congruence.
\end{proof}

We finish this section by noting that Corollary~\ref{corr:main:dyncongr} follows by Corollary~\ref{cor:posit_eqv} and Corollary~\ref{cor:princ_eqv}.
\begin{proofof}{Proof of Corollary 1.3}
Assume that $I$ and $J$ are non-negative connected posets of Dynkin type $\AA_m$. Our aim is to show that $I\approx_\ZZ J$ if and only if $\cox_I(t)=\cox_J(t)$. Since the ``$\Rightarrow$'' implication follows from Fact~\ref{fact:sgc_conseq}\ref{fact:sgc_conseq:coxinvariants}, it suffices to show ``$\Leftarrow$''.

Assume that $\cox_I(t)=\cox_J(t)$. It follows that $|I|=|J|$ and, consequently, $\crk_I=\crk_J$. By Fact~\ref{fact:a:hasse}\ref{fact:a:hasse:crkbiggeri}, we know that either $\crk_I=0$ or $\crk_I=1$. The implication follows by Corollary~\ref{cor:posit_eqv} in the first case and by Corollary~\ref{corr:main:dyncongr} in the second one.
\end{proofof}

\section{Algorithms}\label{sec:algorithms}
Since isomorphic posets are strongly Gram $\ZZ$-congruent~\eqref{eq:isomorphism:stronggram}, we can view $\ZZ$-congruence as a generalization of the notion of isomorphism. Similar to the isomorphism case~\cite{mckayPracticalGraphIsomorphism2014}, we have two general strategies for devising algorithms that construct a matrix $B\in\Gln$ that defines strong Gram $\ZZ$-congruence between two non-negative posets $I$ and $J$ of Dynkin type $\Dyn_I=\Dyn_J=\AA_m$:
\begin{enumerate}[label=\normalfont{(\makebox[\widthof{$ii$}][c]{\roman*})}] 
  \item compute such a $B$ that $I\stackrel{B}{\approx}_\ZZ J$ \textit{directly};
  \item compute a \textit{canonical representative} $R$ and $\ZZ$-congruences $I\stackrel{B_I}{\approx}_\ZZ R$ and $J\stackrel{B_J}{\approx}_\ZZ R$, then calculate $B$ as $B\eqdef B_I\cdot B_J^{-1}$.
\end{enumerate}

In the first strategy, we are given two posets, encoded in, e.g., incidence matrix form \eqref{eq:pos_inc_mat}, and our aim is to construct such a matrix $B\in\Gln$, that $B^{tr}\cdot C_I\cdot B=C_J$. The reader is referred to \cite{gasiorekOnepeakPosetsPositive2012,gasiorekAlgorithmicCoxeterSpectral2020,gasiorekCongruenceRationalMatrices2023} for discussion of some of the algorithms of this type.

The main advantage of the second strategy is its efficiency in a general classification setting defined as follows: divide a (possibly large) list of posets into sublists of $\ZZ$-congruent ones. Since two posets $I$ and $J$ are $\ZZ$-congruent if and only if their canonical representatives are equal, it is more efficient to compute representatives first and then use a sorting, hashing, or
balanced tree algorithm to efficiently group the input list into desired sublists. This methodology is analogous to \textit{canonical labeling} in a graph isomorphism setting, see \cite{mckayPracticalGraphIsomorphism2014}. It is the strategy we follow in the current work. 

By Theorem~\ref{thm:mainthm:stronggram}, statements~\ref{thm:mainthm:stronggram:posit:opcongr} and \ref{thm:mainthm:stronggram:princ:pancongr}, every connected non-negative poset $I$ of Dynkin type $\Dyn_I=\AA_n$ is strongly Gram $\ZZ$-congruent with either the one-peak poset ${}_0\AA_{n-1}^*$ or  a canonical two peak poset ${}_p\wt\AA_{n}$,
\[
\CH({}_0\AA_{n-1}^*)=
\graphOnePeakpan\!\!,\quad\CH({}_p\wt\AA_{n})=\graphOnePeakpanPrinc.
\]
Since $I\approx_\ZZ {}_0\AA_{n-1}^*$ if and only if $\crk_I=0$ and $I\approx_\ZZ {}_p\wt\AA_{n}$ if and only if $\crk_I=1$, we prepare two algorithms for each of these cases. 
{
\makeatletter
\renewcommand{\ALG@name}{Listing}
\makeatother
\begin{algorithm}[H]
    \caption{ \textsc{ReflMatrix()} function}\label{lst:reflmatrix}
    \begin{algorithmic}[1]
\Function{ReflMatrix}{$i, a, b, n$}
    \State $S \gets E_n\in\MM_n(\ZZ)$\Comment{\makebox[66pt]{Identity matrix}}
    \State $S[i,i]\gets -1$;\ $S[a, i]\gets S[b, i]\gets 1$ \Comment{\makebox[66pt]{$S[x, y]\equiv s_{x,y}$\hfill}}
    \State \Return $S$ 
    \EndFunction
\end{algorithmic}
\end{algorithm}}%

Algorithm \ref{alg:refl:posit} is the implementation of the procedure described in the proof of Theorem~\ref{thm:mainthm:stronggram}\ref{thm:mainthm:stronggram:posit:opcongr}. Its main idea is to use the ($min,max$)-reflection operation, implemented in the form of \textsc{ReflMatrix()} function (see Listing~\ref{lst:reflmatrix}), to perform the reduction $J\mapsto {}_0\AA_{n-1}^*$ and construct matrix $B\in\Gln$ as the composition of ($min,max$)-reflection matrices.
We note that this procedure is guaranteed to finish at the poset isomorphic with ${}_0\AA_{n-1}^*$ and not necessarily equal to it. In lines \ref{alg:refl:posit:relabeli}-\ref{alg:refl:posit:relabelii} of the algorithm, we take this into consideration and reorder the resulting matrix to ensure that the  algorithm  returns a matrix that defines the strong Gram $\ZZ$-congruence with the canonically numbered poset ${}_0\AA_{n-1}^*$ \eqref{hasse:pospositzeroanmi}.

\begin{breakablealgorithm}
\caption{Strong Gram $\ZZ$-congruence of a positive connected poset $J$ with $\Dyn_J=\AA_n$ with a one-peak poset ${}_0\AA_{n-1}^*$}\label{alg:refl:posit}
\hspace*{\algorithmicindent}%
\textbf{Input} Incidence matrix $C_J\in\MM_n(\ZZ)$ of a positive connected poset $J$ of the Dynkin type $\AA_n$. \hfill\mbox{}\\
\hspace*{\algorithmicindent}%
\textbf{Output} Nonsingular matrix $B\in\MM_n(\ZZ)$ defining strong Gram $\ZZ$-congruence $J\stackrel{B}{\approx_\ZZ} {}_0\AA_{n-1}^*$.\hfill\mbox{}
\begin{algorithmic}[1]
\Function{StronCongrPositA}{$C_J\in\MM_n(\ZZ)$}
    \State $A \gets C_J^{-1}\in\MM_n(\ZZ)$
    \State $B \gets E_n\in\MM_n(\ZZ)$
    \State $D \gets 0_n\in\MM_n(\ZZ)$\Comment{Zero (null) matrix}
    \State $neigh\gets [\,]$; $is\_not\_peak \gets [\,]$
     \For {$j\in\{0,\ldots,n-1\}$}
     \State $ins\gets \{i;\ a_{i,j}=-1\}$;\ $outs\gets \{i;\ a_{j,i}=-1\}$
     \State $neigh.append(ins \cup outs)$
     \State $is\_not\_peak.append((ins\neq\emptyset)\ \textnormal{\textbf{and}}\ (outs\neq\emptyset))$
    \If {$len(neigh[len(neigh)-1])=1$}\Comment{Vertex $j$ is a \textit{leaf} in $\CH(J)$}
    \State $p\gets j$;\ $p\_is\_min \gets (ins\neq\emptyset)$
    \EndIf
    \EndFor
    \State $line\gets [p, neigh[p].pickElement()]$
    \While{$(nxt\gets line[len(line)-1]\, \setminus\, line[len(line)-2])\neq \emptyset$}
    \State $line.append(nxt.pop())$
    \EndWhile
    \State $k\gets 0$
    \State $C\gets\Call{ReflMatrix}{neigh[0],neigh[1],neigh[1],n}$
    \For{$i\in\{2,\ldots,n-1\}$}\label{alg:refl:posit:mainfor}
    \If{$is\_not\_peak[line[i-1]]$}
    \State \textbf{continue}
    \EndIf
    \State $is\_not\_peak[line[i]]\gets \neg is\_not\_peak[line[i]]$ 
    \For{$j\in\{k+1,\ldots, i-1 \}$}\label{alg:refl:posit:linefor}
    \State $C\gets\Call{ReflMatrix}{line[j],line[j-1],line[j+1],n} \cdot C$\label{alg:refl:posit:linemuli}
    \EndFor
    \State $B\gets B\cdot C$\label{alg:refl:posit:linemul}
    \State $k\gets i-1$
    \EndFor
    \If{$p\_is\_min$}\label{alg:refl:posit:relabeli}
    \State $line\gets line.reverse()$
    \EndIf
     \For{$i\in\{1,\ldots, n \}$}
    \State $D.setColumn(i, B.getColumn(line[i]))$\label{alg:refl:posit:relabelii}
    \EndFor
    \State \Return $D$
    \EndFunction
\end{algorithmic}
\end{breakablealgorithm}
\begin{corollary}\label{cor:positalg:complex}
Algorithm~\ref{alg:refl:posit} performs at most $2(n-2)=O(n)$ matrix multiplications, hence it is of $O(n^4)$ time complexity (assuming na\"{\i}ve matrix multiplication algorithm).
\end{corollary}
\begin{proof}
It is straightforward to see that matrix multiplication is the dominant (i.e., most time-consuming) operation of Algorithm~\ref{alg:refl:posit}. The total number of these operations equals at most 
\begin{itemize}
\item $n-2$ in line \ref{alg:refl:posit:linemuli} and 
\item $n-2$ in line \ref{alg:refl:posit:linemul}.
\end{itemize}
Hence we have at most $2(n-2)=O(n)$ matrix multiplications.
\end{proof}

The number of matrix multiplications performed by Algorithm~\ref{alg:refl:posit} depends on the shape of the Hasse quiver $\CH(J)$. In order to verify that the number $2(n-2)$ is actually attainable, consider a poset $J$ with $\CH(J)$ of the shape
\begin{equation}\label{cor:positalg:complex:pesshasse}
\CH(J)=
\tikzsetnextfilename{hasseposit0Anmipesimistic}\begin{tikzpicture}[baseline=(Alabel.base),label distance=-2pt, xscale=0.8]
        \node[circle, fill=black, inner sep=0pt, minimum size=3.5pt, label={[name=Alabel]left:$\scriptstyle 2$}] (n1) at (0  , 0.5  ) {};
        \node[circle, fill=black, inner sep=0pt, minimum size=3.5pt, label=above:$\scriptstyle 3$] (n2) at (1  , 1  ) {};
        \node (n3) at (2  , 1  ) {};
        \node (n4) at (3  , 1  ) {};
        \node[circle, fill=black, inner sep=0pt, minimum size=3.5pt, label=above:$\scriptstyle n-1$] (n5) at (4  , 1  ) {};
        \node[circle, fill=black, inner sep=0pt, minimum size=3.5pt, label=above:$\scriptstyle \phantom{1}n\phantom{1}$] (n6) at (5  , 1  ) {};
        \node[circle, fill=black, inner sep=0pt, minimum size=3.5pt, label=right:$\scriptstyle 1$] (n0) at (1  , 0  ) {};
        \draw [-stealth, shorten <= -2.50pt, shorten >= 2.50pt] (n4) to  (n5);
        \draw [dotted, -, shorten <= -2.50pt, shorten >= -2.50pt] (n3) to  (n4);
        \foreach \x/\y in {1/0,1/2, 5/6}
        \draw [-stealth, shorten <= 2.50pt, shorten >= 2.50pt] (n\x) to  (n\y);
        \draw [-stealth, shorten <= 2.50pt, shorten >= -2.50pt] (n2) to  (n3);
\end{tikzpicture}
\textnormal{or\ \  }
\CH(J)=
\tikzsetnextfilename{hasseposit0Anmipesimisticii}\begin{tikzpicture}[baseline=(Alabel.base),label distance=-2pt, xscale=0.8]
        \node[circle, fill=black, inner sep=0pt, minimum size=3.5pt, label={[name=Alabel]left:$\scriptstyle 2$}] (n1) at (0  , 0.5  ) {};
        \node[circle, fill=black, inner sep=0pt, minimum size=3.5pt, label=above:$\scriptstyle 3$] (n2) at (1  , 1  ) {};
        \node (n3) at (2  , 1  ) {};
        \node (n4) at (3  , 1  ) {};
        \node[circle, fill=black, inner sep=0pt, minimum size=3.5pt, label=above:$\scriptstyle n-1$] (n5) at (4  , 1  ) {};
        \node[circle, fill=black, inner sep=0pt, minimum size=3.5pt, label=above:$\scriptstyle \phantom{1}n\phantom{1}$] (n6) at (5  , 1  ) {};
        \node[circle, fill=black, inner sep=0pt, minimum size=3.5pt, label=right:$\scriptstyle 1$] (n0) at (1  , 0  ) {};
        \draw [stealth-, shorten <= -2.50pt, shorten >= 2.50pt] (n4) to  (n5);
        \draw [dotted, -, shorten <= -2.50pt, shorten >= -2.50pt] (n3) to  (n4);
        \foreach \x/\y in {1/0,1/2, 5/6}
        \draw [stealth-, shorten <= 2.50pt, shorten >= 2.50pt] (n\x) to  (n\y);
        \draw [stealth-, shorten <= 2.50pt, shorten >= -2.50pt] (n2) to  (n3);
\end{tikzpicture}.
\end{equation}
It is easy to verify that in this case, Algorithm~\ref{alg:refl:posit} performs exactly  $n-2$ steps (line~\ref{alg:refl:posit:mainfor}), and each step performs exactly $2$ matrix multiplications: one in line~\ref{alg:refl:posit:linemuli} and second in line~\ref{alg:refl:posit:linemul}.

\begin{remark}
The reduction procedure described in the proof of Theorem~\ref{thm:mainthm:stronggram}\ref{thm:mainthm:stronggram:posit:opcongr}, applied to poset $J$ of the shape \eqref{cor:positalg:complex:pesshasse}, is as follows:
\begin{equation*}
J_1\eqdef J\mapsto J_2\eqdef S_1S_2 J_1\mapsto J_3\eqdef S_1S_2S_3 J_2\mapsto \cdots \mapsto J_{n-1}\eqdef S_1S_2\cdots S_{n-1} J_{n-2}\simeq {}_0\AA_{n-1}^*.
\end{equation*}
Hence, we can construct matrix $B$ that defines  strong Gram $\ZZ$-congruence $J\stackrel{B}{\approx_\ZZ} {}_0\AA_{n-1}^*$ as:
\[
B\eqdef (B_2 B_1) (B_3B_2 B_1)\cdots(B_{n-1}\cdots B_2 B_1)\in\Gln,
\]
where $B_a=[b^a_{i,j}]\in\MM_n(\ZZ)$ \eqref{eq:reflmat} is a matrix defining ($min,max$)-reflection $S_a$, see Proposition~\ref{prop:minmaxrefl}\ref{cor:minmaxrefl:strongcongr}.
 The straightforward implementation of this procedure yields $\frac{1}{2}(n-2)(n+1)=O(n^2)$ matrix multiplications. In Algorithm~\ref{alg:refl:posit}, we reuse partial results at each step (see lines \ref{alg:refl:posit:linefor}-\ref{alg:refl:posit:linemul}) to significantly reduce the number of multiplications.
\end{remark}

Assume now, that $J=(\{1,\ldots,n\},\preceq)$ is a principal connected poset of Dynkin type $\AA_{n-1}$. By Fact~\ref{fact:a:hasse}\ref{fact:a:hasse:princ}, the Hasse quiver $\CH(J)$ is $2$-regular, and $J$ has at least two maximal elements. It follows that  $\ov\CH(J)$ is a cycle graph, and  $J$ has at least two \textit{minimal} elements.
{
\makeatletter
\renewcommand{\ALG@name}{Listing}
\makeatother

\begin{breakablealgorithm}
    \caption{\textsc{ReflMatrixCirc()}, \textsc{NextOnCircle()} and \textsc{MinMaxMove()} functions}\label{lst:reflmatrixcirc}
    \begin{algorithmic}[1]
    \Function{ReflMatrixCirc}{$u, n, circle$}
    \State $S \gets E_n\in\MM_n(\ZZ)$\Comment{\makebox[66pt]{Identity matrix}}
    \State $a\gets circle[u-1\bmod n]$;\ $i\gets circle[u\bmod n]$;\ $b\gets circle[u+1\bmod n]$
    \State $S[i,i]\gets -1$;\ $S[a, i]\gets S[b, i]\gets 1$ \Comment{\makebox[66pt]{$S[x, y]\equiv s_{x,y}$}}
    \State \Return $S$ 
    \EndFunction\vspace{8pt}
    \Function{NextOnCircle}{$i, n, is\_mmx, circle$}
    \While {$\neg is\_mmx[circle[i\bmod n]]$}
    \State $i \gets i+1$
    \EndWhile
    \State \Return $i \bmod n$ 
    \EndFunction\vspace{8pt}
    \Function{MinMaxMove}{$k, p, is\_mmx, circle$}
    \State $B \gets E_n\in\MM_n(\ZZ)$
    \For{$i\in\{k,k-1,\ldots, p+1 \}$}
    \State $B\gets B\cdot\Call{ReflMatrixCirc}{i, n, circ}$
    \State $is\_minmax[circ[i-1]]\gets \neg is\_minmax[circ[i-1]]$
    \State $is\_minmax[circ[i+1]]\gets \neg is\_minmax[circ[i+1]]$
    \EndFor
    \State \Return $B$
    \EndFunction 
\end{algorithmic}
\end{breakablealgorithm}}%

Algorithm \ref{alg:refl:princ} is the implementation of the procedure described in the proof of Theorem~\ref{thm:mainthm:stronggram}\ref{thm:mainthm:stronggram:princ:pancongr}. Similarly as in case of Algorithm \ref{alg:refl:posit}, its main idea is to use the ($min,max$)-reflection operation (see \textsc{ReflMatrixCirc()} function in Listing~\ref{lst:reflmatrixcirc}) to perform the reduction $J\mapsto {}_p\wt\AA_{n-1}$.
Finally, in lines \ref{alg:refl:princ:relabeli}-\ref{alg:refl:princ:relabelii}, we ensure that  the algorithm  returns a matrix that defines the strong Gram $\ZZ$-congruence with the canonically numbered poset ${}_p\AA_{n}$ \eqref{hasse:pospprincpan}.

\begin{breakablealgorithm}
\caption{Strong Gram $\ZZ$-congruence of a principal connected poset $J$ with $\Dyn_J=\AA_n$ with a canonical two-peak poset ${}_p\wt \AA_{q}$}\label{alg:refl:princ}
\hspace*{\algorithmicindent}%
\textbf{Input} Incidence matrix $C_J\in\MM_n(\ZZ)$ of a principal connected poset $J$ of the Dynkin type $\AA_n$. \hfill\mbox{}\\
\hspace*{\algorithmicindent}%
\textbf{Output} Nonsingular matrix $B\in\MM_n(\ZZ)$ defining strong Gram $\ZZ$-congruence $J\stackrel{B}{\approx_\ZZ} {}_p\wt\AA_{n}$.\hfill\mbox{}
\begin{algorithmic}[1]
    \Function{StronCongrPositA}{$C_J\in\MM_n(\ZZ)$}
    \State $A \gets C_J^{-1}\in\MM_n(\ZZ)$;\ 
    \State $B \gets E_n\in\MM_n(\ZZ)$
    \State $D \gets 0_n\in\MM_n(\ZZ)$\Comment{Zero (null) matrix}
    \State $neigh\gets [\,]$; $is\_mm \gets [\,]$
     \For {$j\in\{0,\ldots,n-1\}$}
     \State $ins\gets \{i;\ a_{i,j}=-1\}$;\ $outs\gets \{i;\ a_{j,i}=-1\}$
     \State $neigh.append(ins \cup outs)$
     \State $is\_mm.append((ins=\emptyset)\ \textnormal{\textbf{or}}\ (outs=\emptyset))$\Comment{List of minimal and maximal vertices}
    \If {$ins=\emptyset$}\Comment{Vertex $j$ is minimal in $J$}
    \State $v\gets j$
    \EndIf
    \EndFor
    \State $circ\gets [v, neigh[v].pop()]$
    \While{$(w \gets (neigh[circ[len(circ) -1]] \setminus set(circ[len(circ)-2])).pop()) \neq v$}
    \State $circ.append(w)$
    \EndWhile
    \For{$i\in\{1,2 \}$}
    \If{$\neg is\_mm[circ[i]]$}
    \State $B\gets B\cdot \Call{MinMaxMove}{\Call{NextOnCircle}{i, n, is\_mm, circ}, i, is\_mm, circ}$
    \EndIf
    \EndFor
    
    \State $ p = \Call{NextOnCircle}{3, n, is\_mm, circ}$
    \While {$(k\gets \Call{NextOnCircle}{p+1, n, is\_mm, circ})\neq 0$ }
    \State $B\gets B\cdot \Call{MinMaxMove}{k, p, is\_mm, circ}$
    \State $ p = \Call{NextOnCircle}{p+1, n, is\_mm, circ}$
    \EndWhile
    \If{$p < n-p+2$}
    \State $\Call{swap}{circ[0], circ[2]}$
    \For{$i\in\{3,\ldots, \lfloor \frac{n-3}{2} \rfloor + 2\}$}
    \State $\Call{swap}{circ[i], circ[n-i+2]}$
    \EndFor
    \State $p\gets n-p+2$
    \EndIf
    
    \State $perm\gets [0,\ldots,0]\in \ZZ^n$\Comment{$n$-elemented array}\label{alg:refl:princ:relabeli}
    \State $perm[0]\gets circ[2]$;\ $perm[1]\gets circ[0]$
    \State $perm[n-1]\gets circ[p]$;\ $perm[n-2]\gets circ[1]$
    \For{$i\in\{3,\ldots, p - 1\}$}
    \State $perm[i-1]\gets circ[i])$
    \EndFor
    \For{$i\in\{1,\ldots, n - p - 1\}$}
    \State $perm[p+i-2]\gets circ[n-i])$
    \EndFor
     \For{$i\in\{1,\ldots, n \}$}
    \State $D.setColumn(i, B.getColumn(perm[i]))$\label{alg:refl:princ:relabelii}
    \EndFor
    \State \Return $D$
    \EndFunction
\end{algorithmic}
\end{breakablealgorithm}
\begin{corollary}\label{cor:princalg:complex}
Algorithm~\ref{alg:refl:princ} performs at most $\lfloor\frac{n-2}{2}\rfloor\cdot\lceil\frac{n-2}{2}\rceil=O(n^2)$ matrix multiplications, hence it is of $O(n^5)$ time complexity (assuming na\"{\i}ve matrix multiplication algorithm).
\end{corollary}
\begin{proof}
It is straightforward to see that matrix multiplication is the dominant (i.e., most time-consuming) operation of Algorithm~\ref{alg:refl:princ}. Assume that $J$ is a principal connected poset of Dynkin type $\AA_{n-1}$ and 

\begin{equation}\label{cor:princalg:complex:pesshasse}
\CH(J)=
\tikzsetnextfilename{hasseprincpesimistic}\begin{tikzpicture}[baseline=(Alabel.base),label distance=-2pt, xscale=0.8, yscale=0.6]
\node[circle, fill=black, inner sep=0pt, minimum size=3.5pt, label=above:$\scriptstyle 1$] (n1) at (6  , 2  ) {};
\node[circle, fill=black, inner sep=0pt, minimum size=3.5pt, label=above:$\scriptstyle 2$] (n2) at (5  , 2  ) {};
\node[circle, fill=black, inner sep=0pt, minimum size=3.5pt, label=above:$\scriptstyle 3$] (n3) at (4  , 2  ) {};
\node (n4) at (3  , 2  ) {$\scriptstyle $};
\node (n5) at (2  , 2  ) {$\scriptstyle $};
\node[circle, fill=black, inner sep=0pt, minimum size=3.5pt, label=above:$\scriptstyle k-1$] (n6) at (1  , 2  ) {};
\node[circle, fill=black, inner sep=0pt, minimum size=3.5pt, label={[name=Alabel]left:$\scriptstyle k$}] (n7) at (0  , 1  ) {};
\node[circle, fill=black, inner sep=0pt, minimum size=3.5pt, label=below:$\scriptstyle k+1$] (n8) at (1  , 0  ) {};
\node (n9) at (2  , 0  ) {$\scriptstyle $};
\node (n10) at (4  , 0  ) {$\scriptstyle $};
\node[circle, fill=black, inner sep=0pt, minimum size=3.5pt, label=below:$\scriptstyle n-2$] (n11) at (5  , 0  ) {};
\node[circle, fill=black, inner sep=0pt, minimum size=3.5pt, label=below:$\scriptstyle n-1$] (n12) at (6  , 0  ) {};
\node[circle, fill=black, inner sep=0pt, minimum size=3.5pt, label=left:$\scriptstyle n$] (n13) at (5  , 1  ) {};
\foreach \x/\y in {2/1, 3/2, 7/6, 7/8, 11/12, 13/1, 13/12}
        \draw [-stealth, shorten <= 2.50pt, shorten >= 2.50pt] (n\x) to  (n\y);
\foreach \x/\y in {6/5, 8/9}
        \draw [-stealth, shorten <= 2.50pt, shorten >= -2.50pt] (n\x) to  (n\y);
\foreach \x/\y in {5/4, 9/10}
        \draw [line width=1.2pt, line cap=round, dash pattern=on 0pt off 5\pgflinewidth, -, shorten <= -2.50pt, shorten >= -2.50pt] (n\x) to  (n\y);
\foreach \x/\y in {4/3, 10/11}
        \draw [-stealth, shorten <= -2.50pt, shorten >= 2.50pt] (n\x) to  (n\y);
\end{tikzpicture}\!\!\!,\textnormal{ where }
k=\begin{cases}
\frac{n}{2},&\textnormal{ if $n$ is even,}\\
\frac{n-1}{2},&\textnormal{ if $n$ is odd.}
\end{cases}
\end{equation}
One checks that:\begin{itemize}
\item 
the  number of matrix multiplication performed by
Algorithm~\ref{alg:refl:princ} applied to the poset $J$ equals:
\[
\left\lfloor\tfrac{n-2}{2}\right\rfloor\cdot\left\lceil\tfrac{n-2}{2}\right\rceil=
\begin{cases}
(k-1)^2,&\textnormal{ if $n$ is even,}\\
k(k-1),&\textnormal{ if $n$ is odd;}
\end{cases}
\]
\item poset $J$, up to the isomorphism, represents the pessimistic case that requires the most matrix multiplications.
\end{itemize}
The details are left to the reader.
\end{proof}

\section{Enumeration of Dynkin type \texorpdfstring{$\AA_m$}{Am} posets}

Throughout this section, we assume that $I$ is a non-negative poset of a Dynkin type $\AA_m$. Our aim is to give formulas for the exact number of all (up to the isomorphism) posets $I$ with a given Coxeter polynomial (called \textit{Coxeter type}).
The total number of such non-negative posets $I$ is given in \cite{gasiorekStructureNonnegativePosets2022arxiv} and is as follows.

{\sisetup{group-minimum-digits=4}%
\begin{fact}\label{fact:typeanum}
Let $Nneg(\AA_n)$ be the number of all non-negative posets $I$ of size $n=|I|$ and Dynkin type $\Dyn_I=\AA_{n-\crk_I}$. Then
\begin{equation}\label{thm:typeanum:eq}
Nneg(\AA_n)=
\begin{cases}
1 & \textnormal{ if }n\in\{1,2\},\\
2^{n - 2} + \frac{1}{2n} \sum_{d\mid n}\left(2^{\frac{n}{d}}\varphi(d)\right) + 2^{\frac{n - 3}{2}}-\lceil\frac{n+1}{2}\rceil, & \textnormal{ if } n\geq 3 \textnormal{ is odd},\\[0.1cm]
2^{n - 2} + \frac{1}{2n} \sum_{d\mid n}\left(2^{\frac{n}{d}}\varphi(d)\right) + 2^{\frac{n}{2}-2}-\lceil\frac{n+1}{2}\rceil, & \textnormal{ if } n\geq 4 \textnormal{ is even},\\
\end{cases}
\end{equation}
where $\varphi$ is Euler's totient function. Moreover, there are exactly
\begin{equation}\label{fact:digrphnum:path:eq}
N(P_n)=
\begin{cases}
2^{n-2}, & \textnormal{if $n\geq 2$ is even},\\[0.1cm]
2^{\frac{n - 3}{2}} + 2^{n - 2}, & \textnormal{if $n\geq 1$ is odd,}\\
\end{cases}
\end{equation}
positive posets, and $Nneg(\AA_n)-N(P_n)$ principal ones.
\end{fact}}%

In view of Theorem~\ref{thm:mainthm:stronggram}, to obtain the number of posets $I$ of a given Coxeter polynomial $\cox_I(t)\in\ZZ[t]$ \eqref{eq:pos_cox_poly}, we count oriented paths and (some) oriented cycles. Indeed%
\begin{enumerate}[label=\normalfont{(\roman*)}]
\item\label{enum:a:posit} $I$ is positive and $\cox_I(t) = t^{n+1}+t^n+\ldots + t + 1\in\ZZ[t]$ iff $\CH(I)$ is an oriented path,
\item\label{enum:a:princ} $I$ is principal and $\cox_I(t) = t^n- t^p-t^{n-p}+1\in\ZZ[t]$ iff $\CH(I)$ is such an oriented cycle,
that
\begin{itemize}
\item there are at least two sinks,
\item there are exactly $p$ edges oriented clockwise, see Remark~\ref{rmk:cycle_index_hasse}.
\end{itemize}
\end{enumerate}

Since there are exactly $N(P_n)$ \eqref{fact:digrphnum:path:eq} oriented paths (see \cite{gasiorekStructureNonnegativePosets2022arxiv}), it is sufficient to count oriented cycles described in~\ref{enum:a:princ}. We start with counting such oriented cycles $D=(V_{D},E_{D})$, that have exactly $n\eqdef |V_D|\geq 3$ vertices and $p\geq \frac{n}{2}$ edges oriented clockwise. 

Assume that $D=(\{1,\ldots,n\}, E_{D}\})$ is an oriented cycle. Without loss of generality, we may assume that $D$
is depicted in the circle layout on the Euclidean plane and its edges $e\in E_{D}$ are labeled with two colors:\pagebreak

\begin{itemize}
\item \textit{black}: if $e$ is clockwise oriented, and
\item \textit{white}: if $e$ is counterclockwise oriented.
\end{itemize} 
This way, every cycle $D$ can be viewed as a binary combinatorial necklace $\CN_2(n)$. We recall that a binary necklace of length $n$ is an equivalence class of binary strings $s\in\{0,1\}^n$ on length $n$, with all \textit{rotations} considered equivalent. By a \textit{rotation} we mean a \textit{circular shift}, i.e., a cyclic permutation consisting of a single nontrivial cycle. In other words, a binary combinatorial necklace represents a structure of $n$ circularly connected beads, each of which is either black or white.

\begin{proposition}\label{prop:binneckl:num}
The number $N\!ck_2(p,n)$ of binary combinatorial necklaces that have exactly $p$ black and $n-p$ white beads: 
\begin{enumerate}[label=\normalfont{(\alph*)}]
 \item\label{prop:binneckl:num:genfunc} is given by the coefficient $a_{p,n-p}\in\ZZ$ of the monomial $b^pw^{n-p}$ in the generating function
\begin{equation}\label{prop:binneckl:num:genfunc:eq}
G_{\CN_2(n)}(b,w)=\sum_{p=0}^n a_{p,n-p} b^pw^{n-p}= \frac{1}{n} \sum_{d|n} (b^{n/d} + w^{n/d})^d \varphi \left( \frac{n}{d} \right)\in\ZZ[b,w];
\end{equation}
 \item\label{prop:binneckl:num:exact} equals:
 \begin{equation}
N\!ck_2(p,n)=\frac{1}{n}\sum_{d|\gcd(n,p)}\varphi(d)\binom{n/d}{p/d}.
\end{equation} 
\end{enumerate}
\end{proposition}
\begin{proof}
Part \ref{prop:binneckl:num:genfunc} of the proposition follows directly from the Pólya enumeration theorem \cite{polyaCombinatorialEnumerationGroups1987} and to show \ref{prop:binneckl:num:exact} we use the Burnside's lemma (also known as the Cauchy–Frobenius lemma) as follows.

Assume that $X$ is the set of binary necklaces with $n$ beads, $p$ of which is black, and $G$ denotes the rotations group. By Burnside's lemma, we have that
\begin{equation*}
N\!ck_2(p,n)=\frac{1}{|G|}\sum_{g\in G} |X^g|,
\end{equation*}
where $X^g$ denotes the set on necklaces fixed by $g\in G$, i.e. $X^g\eqdef \{x\in X;\, gx=x\}$. Now, we fix an element $g\in G$ and describe necklaces $x\in X^g$. Note that $|G|=n$, since every rotation $g\in G$ can be viewed as a clockwise rotation by $0\leq k < n$ beads. It follows that every $x\in X^g$ consists of $d\eqdef \frac{n}{k}$ identical \textit{linear} (continuous) sections of length $k= \frac{n}{d}$. These sections can be chosen in
\begin{equation*}
\binom{n/d}{p/d,(n-p)/d}=\frac{\frac{n}{d}!}{\frac{p}{d}!\frac{n-p}{d}!}=\binom{n/d}{p/d}
\end{equation*}
ways and we conclude that $X^g\neq \emptyset$ if and only if $d|\gcd(n,p)$. Moreover, for each divisor $d$ of $\gcd(n,p)$, there are $\varphi(d)$ rotations with the same number of symmetries as rotation by $k$ beads, that~is:
\begin{equation*}
N\!ck_2(p,n)=\frac{1}{|G|}\sum_{g\in G} |X^g|=\frac{1}{n}\sum_{d|\gcd(n,p)}\varphi(d)\binom{n/d}{p/d},
\end{equation*}
where $\varphi(d)\eqdef|\{1\leq k<n;\,\gcd(k,n)=1 \}|$ is Euler's totient function.
\end{proof}

\begin{example}
Consider the set $\CN_2(10)$ of all binary necklaces consisting of $10$ beads. 
The generating function $G(b,w)\eqdef G_{\CN_2(10)}(b,w)\in\ZZ[b,w]$ is given by the formula
\begin{equation*}
G(b,w)=b^{10} + b^{9} w + 5 b^{8} w^{2} + 12 b^{7} w^{3} + 22 b^{6} w^{4} + 26 b^{5} w^{5} + 22 b^{4} w^{6} + 12 b^{3} w^{7} + 5 b^{2} w^{8} + b w^{9} + w^{10}\!.
\end{equation*}
Hence, the total number of such necklaces equals $|\CN_2(10)|=G_{\CN_2(10)}(1,1)=108$. In particular, there exists exactly
 \begin{equation*}
N\!ck_2(6,10)=\frac{1}{10}\!\sum_{d|\gcd(10,6)}\!\!\!\varphi(d)\binom{n/d}{p/d}=\frac{1}{10}\left(\varphi(1)\binom{10}{6}+\varphi(2)\binom{5}{3}\right)=\frac{1}{10}(210+10)=22
\end{equation*} 
necklaces that have $6$ black beads. Obviously, this number coincides with the number of necklaces that have exactly $6$ white beads and with the coefficient $a_{6,4}$ of $b^6w^4$ monomial in the $G(b,w)\in\ZZ[b,w]$ generating function \eqref{prop:binneckl:num:genfunc:eq}.
\end{example}

Assume that $C_n=(\{1,\ldots, n\},E_{C_n})$ is an oriented cycle. Since $C_n$ is a planar digraph, we may assume that it is depicted on the Euclidean plane in such a way that no two edges intersect, for example, in the circle layout.
By a cycle index $c(C_n)\in\{\lceil\frac{n}{2} \rceil, \ldots, n \}$ we mean $c(C_n)\eqdef \max(r_{C_n},l_{C_n})$, where $r_{C_n}$ ($l_{C_n}$) denotes the number of clockwise (counterclockwise) oriented edges in $C_n$. It is straightforward to see that the cycle index is invariant under digraph isomorphism.

\begin{proposition}\label{prop:orietedcycle:idx:num}
Number $N(C_n,p)$ of all, up to digraph isomorphism, oriented cycles $D$ with the cycle index $c(D)$ equal $p$ is given by the formulae
\begin{equation}\label{prop:orietedcycle:idx:num:eq}
N(C_n,p)=
\begin{cases}
\frac{1}{n}\sum_{d|\gcd(n,p)}\varphi(d)\binom{n/d}{p/d}, & \textnormal{if }p\in\{\lceil\frac{n}{2} \rceil, \ldots, n \}\textnormal{ and } p\neq\frac{n}{2},\\[0.1cm]
\frac{1}{2n}\sum_{d|\gcd(n,\frac{n}{2})}\varphi(d)\binom{n/d}{n/2d}+2^{\frac{n}{2}-2}, & \textnormal{if }p=\frac{n}{2}.\\
\end{cases}
\end{equation}
\end{proposition}
\begin{proof}
We use the combinatorial necklace model. If $p\neq\frac{n}{2}$, it is straightforward to see that $N(C_n,p)$ coincides with the number of combinatorial necklaces that have exactly $p>n-p$ black beads. Hence, the first part of equality \eqref{prop:orietedcycle:idx:num:eq} follows from Proposition~\ref{prop:binneckl:num}\ref{prop:binneckl:num:exact}. To prove the second part, we assume that $p=\frac{n}{2}=n-p$, i.e., we consider combinatorial necklaces that have an equal number of white and black beads. We have two possibilities: we can identify black beads with clockwise or counterclockwise edges. Since switching beads color yields isomorphic necklaces in exactly $2^{\frac{n}{2}-1}$ cases, we conclude that $N(C_n,\frac{n}{2})=2^{\frac{n}{2}-1} + (N\!ck_2(\frac{n}{2},n) -2^{\frac{n}{2}-1})/2 $ and equality \eqref{prop:orietedcycle:idx:num:eq} follows.
\end{proof}

As discussed earlier in this section, every non-negative poset $I$ of a Dynkin type $\Dyn_I=\AA_m$ can be viewed as a binary combinatorial necklace. Moreover, the cycle index $c(I)$ \eqref{eq:df:cycleindex} has a natural interpretation in this model: it is the number of black beads. By combining the results of Theorem~\ref{thm:mainthm:stronggram} and Proposition~\ref{prop:orietedcycle:idx:num} we get the proof of Theorem~\ref{thr:posnum:a}.

\begin{proofof}{Proof of Theorem 1.2}
Apply the results of Fact~\ref{fact:typeanum}, Theorem~\ref{thm:mainthm:stronggram} and Proposition~\ref{prop:orietedcycle:idx:num}. In particular, to prove \ref{thr:posnum:a:printc}, assume that $n\geq 4$ and $\frac{n}{2}\leq p\leq n-2$ are fixed and note that there exists only one poset $I$ (oriented cycle $C_n$) that has exactly one maximal element (one sink) and $c(I)=c(C_n)=p$.
\end{proofof}

\begin{example}
For $n=34$, there exist exactly $\num{4547647110}$ non-isomorphic connected posets $I$ of size $n=|I|$, that are non-negative of Dynkin type $\Dyn_I=\AA_{n-\crk_I}$. Moreover, these posets yield $\lfloor\frac{n}{2}\rfloor=17$ Coxeter types (Coxeter polynomials). By Theorem~\ref{thr:posnum:a}, there are exactly $\num{4294967296}$ such positive ($\crk_I=0$) posets with $\cox_I=t^{34}+t^{33}+t^{32}+\cdots+t+1\in\ZZ[t]$ and $\num{252679814}$ principal ($\crk_I=1$) ones, as described in Table~\ref{tbl:exprinccoxtypes}.
\end{example}

%
\makeatother%
\endgroup%
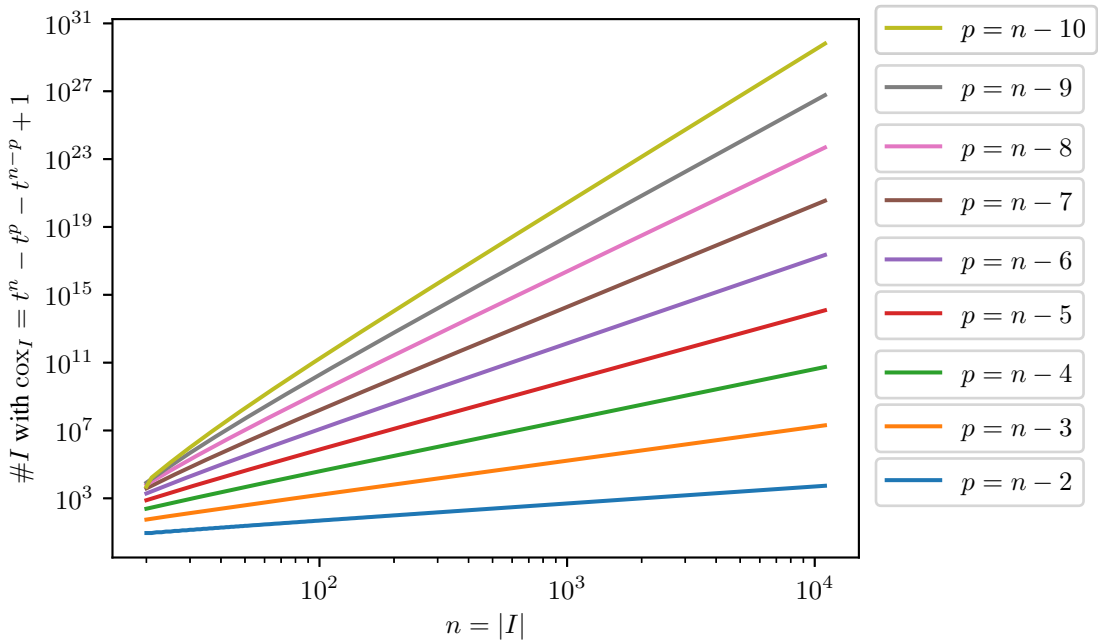
        \caption{Log--log plot of $\#I\textnormal{ with }\mathrm{cox}_I=t^n- t^p-t^{n-p}+1=N(C_n,p)-1$}
        \label{fig:coxgrow}
\end{figure}
\begin{remark}
We recall from \cite{gasiorekStructureNonnegativePosets2022arxiv} that the number of connected non-negative posets $I$ of Dynkin type $\AA_{n-\crk_I}$ grows exponentially when $n\to\infty$. Theorem~\ref{thr:posnum:a}\ref{thr:posnum:a:printc} gives a description of the number $N(C_n,p)-1$ of principal ($\crk_I=0$) posets $I$ with $c(I)=p$. In particular, $N(C_n,p)=\lfloor\frac{n}{2}\rfloor$ for $p=n-2$, i.e., in this case the number of posets grows \textit{linearly}. Otherwise, for $p\neq n-2$, the growth rate is \textit{exponential}, see Figure~\ref{fig:coxgrow} for illustration. 
\end{remark}

\section{Summary and future work}\label{sec:conclusions}
In the present work, we give a complete Coxeter spectral classification of connected non-negative posets $I$ of Dynkin type $\Dyn_I=\AA_{|I|-\crk_I}$. A complete description and classification of Dynkin type $\DD_n$ posets is unknown (partial results are given in \cite{gasiorekOnepeakPosetsPositive2012,gasiorekAlgorithmicStudyNonnegative2015,gasiorekCoxeterTypeClassification2019}). Therefore, the following problem remains open.
\begin{problem}
Give a structural description and Coxeter spectral classification of $\DD_n$ Dynkin type non-negative connected posets.
\end{problem}

In the case of $\AA_m$ type posets, we show in the present work that there are exactly $\lfloor\frac{n}{2}\rfloor$ Coxeter types:
\begin{itemize}
\item $\cox_I(t) = t^{n}+t^{n-1}+\ldots + t + 1\in\ZZ[t]$ if $I$ is positive ($\crk_I=0$),
\item $\cox_I(t) = t^n- t^p-t^{n-p}+1\in\ZZ[t]$, where $\frac{n}{2}\leq p\leq n-2$, if $I$ is principal ($\crk_I=1$).
\end{itemize}
Moreover, we prove that, given a pair of non-negative posets $I$ and $J$ of Dynkin type $\Dyn_I=\Dyn_J=\AA_{|I|-\crk_I}$, the incidence matrices $C_I$ and $C_J$ are $\ZZ$-congruent if and only if $\specc_I = \specc_J$. It is known that the assumption that Dynkin types coincide cannot be omitted in general (see~\cite[Example~4.1]{gasiorekAlgorithmicStudyNonnegative2015}). Our experiments show that this is possible for \textit{small} posets.
\begin{proposition}
If $n\leq 8$ and $I$ is such a poset that $\cox_I(t) = t^n+t^{n-1}+\cdots +t + 1$, then $I$ is connected, positive, and $\Dyn_I=\AA_n$.
\end{proposition}
\begin{proof}
The proof is a computational one. We list all, up to isomorphism, posets $I$ of size $|I|\leq 8$, and verify that posets $J$ that have their Coxeter polynomial $\cox_J(t)\in\ZZ[t]$ \eqref{eq:pos_cox_poly} equal $\cox_J(t) = t^n+t^{n-1}+\cdots +t + 1$ are exactly non-negative ones, of Dynkin type $\Dyn_J=\AA_n$. 
\end{proof}

\begin{example}
Let $I=(\{1,\ldots,9\},\preceq)$ be a  connected poset defined by the following Hasse quiver.
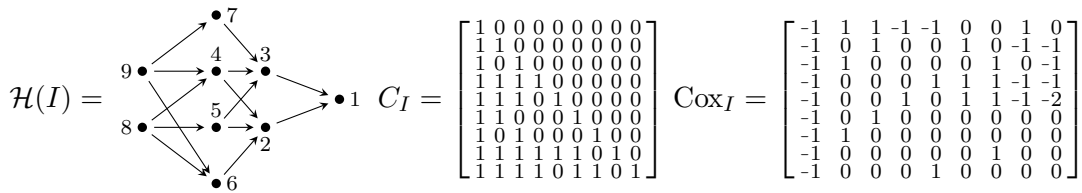
\begin{equation*}
\CH(I)=
\tikzsetnextfilename{exindefcoxa}
\begin{tikzpicture}[baseline={([yshift=-2.75pt]current bounding box)},label distance=-2pt,xscale=0.65, yscale=0.74]
\node[circle, fill=black, inner sep=0pt, minimum size=3.5pt, label=right:$\scriptstyle 1$] (n1) at (4  , 1.50) {};
\node[circle, fill=black, inner sep=0pt, minimum size=3.5pt, label=below:$\scriptstyle 2$] (n2) at (2.50, 1  ) {};
\node[circle, fill=black, inner sep=0pt, minimum size=3.5pt, label=above:$\scriptstyle 3$] (n3) at (2.50, 2  ) {};
\node[circle, fill=black, inner sep=0pt, minimum size=3.5pt, label=above:$\scriptstyle 4$] (n4) at (1.50, 2  ) {};
\node[circle, fill=black, inner sep=0pt, minimum size=3.5pt, label=above:$\scriptstyle 5$] (n5) at (1.50, 1  ) {};
\node[circle, fill=black, inner sep=0pt, minimum size=3.5pt, label=right:$\scriptstyle 6$] (n6) at (1.50, 0  ) {};
\node[circle, fill=black, inner sep=0pt, minimum size=3.5pt, label=right:$\scriptstyle 7$] (n7) at (1.50, 3  ) {};
\node[circle, fill=black, inner sep=0pt, minimum size=3.5pt, label=left:$\scriptstyle 8$] (n8) at (0  , 1  ) {};
\node[circle, fill=black, inner sep=0pt, minimum size=3.5pt, label=left:$\scriptstyle 9$] (n9) at (0  , 2  ) {};
\foreach \x/\y in {2/1, 3/1, 4/2, 4/3, 5/2, 6/2, 7/3, 8/4, 8/5, 8/6, 9/4, 9/7}
        \draw [-stealth, shorten <= 2.50pt, shorten >= 2.50pt] (n\x) to  (n\y);
\draw [-stealth, shorten <= 2.50pt, shorten >= 2.50pt] (n5) to  (n3);
\draw [-stealth, shorten <= 2.50pt, shorten >= 2.50pt] (n9) to  (n6);
\end{tikzpicture}\ C_I=
\begin{bsmallmatrix*}[r]
1 & 0 & 0 & 0 & 0 & 0 & 0 & 0 & 0\\
1 & 1 & 0 & 0 & 0 & 0 & 0 & 0 & 0\\
1 & 0 & 1 & 0 & 0 & 0 & 0 & 0 & 0\\
1 & 1 & 1 & 1 & 0 & 0 & 0 & 0 & 0\\
1 & 1 & 1 & 0 & 1 & 0 & 0 & 0 & 0\\
1 & 1 & 0 & 0 & 0 & 1 & 0 & 0 & 0\\
1 & 0 & 1 & 0 & 0 & 0 & 1 & 0 & 0\\
1 & 1 & 1 & 1 & 1 & 1 & 0 & 1 & 0\\
1 & 1 & 1 & 1 & 0 & 1 & 1 & 0 & 1
\end{bsmallmatrix*}\ \Cox_I=
\begin{bsmallmatrix*}[r]
\shortminus 1 & \phantom{\shortminus}1 & \phantom{\shortminus}1 & \shortminus 1 & \shortminus 1 & \phantom{\shortminus}0 & \phantom{\shortminus}0 & 1 & 0\\
\shortminus 1 & 0 & 1 & 0 & 0 & 1 & 0 & \shortminus 1 & \shortminus 1\\
\shortminus 1 & 1 & 0 & 0 & 0 & 0 & 1 & 0 & \shortminus 1\\
\shortminus 1 & 0 & 0 & 0 & 1 & 1 & 1 & \shortminus 1 & \shortminus 1\\
\shortminus 1 & 0 & 0 & 1 & 0 & 1 & 1 & \shortminus 1 & \shortminus 2\\
\shortminus 1 & 0 & 1 & 0 & 0 & 0 & 0 & 0 & 0\\
\shortminus 1 & 1 & 0 & 0 & 0 & 0 & 0 & 0 & 0\\
\shortminus 1 & 0 & 0 & 0 & 0 & 0 & 1 & 0 & 0\\
\shortminus 1 & 0 & 0 & 0 & 1 & 0 & 0 & 0 & 0
\end{bsmallmatrix*}
\end{equation*}
It is straightforward to check that:
\begin{itemize}
\item $q_I([2, -1, -2, -2, 0, -3, -1, 3, 2])=-3$, i.e., poset $I$ is indefinite (is not non-negative),
\item $\cox_I(t)=t^{9} + t^{8} + t^{7} + t^{6} + t^{5} + t^{4} + t^{3} + t^{2} + t + 1.$
\end{itemize}
\end{example}
In other words, the Coxeter spectrum (Coxeter polynomial) does not preserve the positivity of posets. Nevertheless, the following conjecture holds for posets of at most $13$ elements.

\begin{conjecture}
If $I$ is such a poset, that $\cox_I(t) = t^n+t^{n-1}+\cdots +t + 1$, then
$I$ is connected. If, additionally, $I$ is non-negative, then $\Dyn_I=\AA_{|I|}$.
\end{conjecture}

\end{document}